\documentclass{article}

\usepackage[utf8]{inputenc}
\usepackage{listings}
\usepackage{color}
\usepackage{fullpage}
\usepackage{color}
\usepackage{amssymb,graphicx}
\usepackage{float}
\usepackage{verbatim}
\floatstyle{plain}
\restylefloat{figure}
\usepackage{caption}
\usepackage{subcaption}
\usepackage{tabu}
\usepackage{graphicx}
\usepackage{hyperref}
\usepackage{amsthm}
\usepackage{amsmath}
\usepackage[linesnumbered,lined,boxed,commentsnumbered,ruled]{algorithm2e}
\usepackage{tikz}
\usepackage{makecell}
\usepackage{comment}
\usepackage{mathtools}
\usepackage{etoolbox}
\usepackage{url}
\makeatletter
\newcommand\footnoteref[1]{\protected@xdef\@thefnmark{\ref{#1}}\@footnotemark}
\makeatother
\theoremstyle{definition}
\newtheorem{invariant}{Invariant}
\theoremstyle{definition}

\newtheorem{theorem}{Theorem}[subsection]
\newtheorem{lemma}{Lemma}[subsection]
\theoremstyle{definition}
\newtheorem{definition}{Definition}[subsection]
\newcommand{\bDiamond}{\mathbin{\Diamond}}

\newcommand{\PROTNAME}{FireLedger}  % TOY
\newcommand{\SYSNAME}{FLO} % TOP
\newcommand{\FULLDEF}{\emph{Blockchain Based Finality Consensus}}
\newcommand{\ACRDEF}[1]{BBFC(#1)}
\newcommand{\PREFDEF}[1]{BBFC}

\newtoggle{VLDB}
\newtoggle{MYACM}
\toggletrue{MYACM}
\newtoggle{Annon}
\newtoggle{SC}
\toggletrue{SC}
\newtoggle{FV}
\toggletrue{FV}
\newtoggle{SINGLE}
\toggletrue{SINGLE}

\date{}
\title{\PROTNAME: A High Throughput Blockchain Consensus Protocol
\thanks{This work was partially funded by ISF grant \#1505/16.}}

\author{Yehonatan Buchnik \ \ \ \ \ \ 
	Roy Friedman\\
	Computer Science Department\\
	Technion\\
	\texttt{E-mail: \{yon\_b,roy\}@cs.technion.ac.il}}

\begin{document}
\maketitle

\begin{abstract}
Blockchains are distributed secure ledgers to which transactions are issued continuously and each block of transactions is tightly coupled to its predecessors.
Permissioned block\-chains place special emphasis on transactions throughput.
In this paper we present \PROTNAME, which leverages the iterative nature of blockchains in order to improve their throughput in optimistic execution scenarios.
\PROTNAME{} trades latency for throughput in the sense that in \PROTNAME{} the last $f+1$ blocks of each node's blockchain are considered tentative, i.e., they may be rescinded in case one of the last $f+1$ blocks proposers was Byzantine.
Yet, when optimistic assumptions are met, a new block is decided in each communication step, which consists of a proposer that sends only its proposal and all other participants are sending a single bit each.
Our performance study demonstrates that in a single Amazon data-center, \PROTNAME{} running on $10$ \iftoggle{FV}{mid-range Amazon}{} nodes obtains a throughput of up to $160$K \iftoggle{FV}{transactions per second}{tps} for (typical Bitcoin size) $512$ bytes transactions.
In a $10$ nodes Amazon geo-distributed setting with $512$ bytes transactions, \PROTNAME{} obtains a throughput of $30$K tps.
Moreover, on higher end Amazon machines, \PROTNAME{} obtains $20\%-600\%$ better throughput than state of the art protocols like HotStuff and BFT-SMaRt, depending on the exact configuration.
\end{abstract}
	
	\section{Introduction}

\iftoggle{VLDB}
{
	\emph{Blockchains} are becoming popular in many areas such as cryptocurrencies, supply-chains, insurance, and others~\cite{PwC-survey}.
}
{
A \emph{blockchain} is a distributed secure replicated ledger service designed for environments in which not all nodes can be trusted~\cite{nakamoto}.
Specifically, a blockchain maintains a distributed ordered list of blocks (the ``chain'') in which every block contains a sequence of transactions as well as authentication data about the previous blocks; the latter typically relies on cryptographic methods.
%}{
%A \emph{blockchain} maintains a distributed ordered list of blocks (the ``chain'') in which every block contains a sequence of %transactions as well as authentication data about the previous blocks; the latter typically relies on cryptographic methods.
%}
%Assuming the authentication data cannot be forged,
Hence, any attempt to modify part of the blockchain can be detected, which helps to ensure the stability and finality of blockchain prefixes.
Notice that transactions may in fact be any deterministic computational step, including the execution of \emph{smart contracts} code.
A primary challenge in implementing a blockchain abstraction is deciding on the order of transactions in the chain.

}
\iftoggle{VLDB}
{
	Blockchains are often characterized as either \emph{unpermissioned} or \emph{permissioned}.
	In permissioned mode, the blockchain is executed among a set of $n$ known participants under the assumption that at most $f$ of them are faulty~\cite{cachin-survey}.
}
{
Largely speaking, blockchains can be characterized as either \emph{unpermissioned} or \emph{permissioned}.
In unpermissioned blockchains, any node is allowed to participate in the computational task of deciding the ordering of transactions and blocks as well as in the task of maintaining the blockchain's state~\cite{ethereum,nakamoto}.
In such an environment, there is no trust between nodes and in fact, with crypto-currencies, participants have an a-priori incentive to cheat.
This implies utilizing significant cryptographic mechanisms to compensate for this zero trust model, which limits the system's throughput.

In contrast, in permissioned mode, the blockchain is executed among a set of $n$ known participants under the assumption that at most $f$ of them are faulty~\cite{cachin-survey}.
For example, consider a consortium of insurance companies each donating a node in order to maintain a common blockchain of insurance policies and insurance claims.
}
In this setting, blockchain becomes a special case of traditional \emph{replication state machine} (RSM)~\cite{Lamport1978,Schneider1990}.
A common approach to implementing RSM is by repeatedly running a consensus protocol to decide on the next transaction to be executed~\cite{Lamport1998} with the optimization of batching multiple transactions in each invocation of the consensus protocol~\cite{FriedmanR97}.

The assumed possible type of failures affects the type of consensus protocols that are used.
\emph{Benign} failures such as a node crash and occasional message omission can be overcome by benign consensus protocols, e.g.,~\cite{Chandra1996, MR-2002, Lamport1998, Ongaro:2014}.
On the other hand, \emph{Byzantine} failures~\cite{Lamport1982} in which a faulty node may arbitrarily deviate from its code require Byzantine fault tolerant protocols (BFT), e.g.,~\cite{Guerraoui2010,Castro1999,Kotla2007}.
\nottoggle{VLDB}{
The type of failures along with the synchrony assumptions about the environment in which the blockchain is implemented impact the minimal ratio between $f$ and $n$ required to enable a solution to the problem.
For example, assuming a synchronous environment, benign consensus requires $f<n$ while Byzantine consensus requires $f<n/2$ in the signed messages model, and $f<n/3$ in the oral messages model~\cite{Lamport1982}.
}{}
In the totally asynchronous case, the seminal FLP result showed that even benign consensus cannot be solved~\cite{Fischer1985}.
Yet, when enriching the environment with some minimal eventual synchrony assumptions, e.g., partial synchrony~\cite{Dolev1987,Dwork1988}, or with unreliable failure detector oracles~\cite{Chandra96}, benign consensus becomes solvable when $f<n/2$ and Byzantine consensus requires $f<n/3$.
This is as long as the network does not become partitioned~\cite{FB96}.

Here we focus on permissioned blockchains assuming Byzantine failures and partial synchrony.
\nottoggle{VLDB}{
According to a recent survey by PwC~\cite{PwC-survey}, only}{
A recent survey by PwC~\cite{PwC-survey} indicates that only} a third of companies \iftoggle{FV}{currently}{} using or planing to use blockchains intend on using unpermissioned blockchains.
Also, many recent unpermissioned proposals, e.g., Algorand~\cite{Gilad:2017}, Tendermint~\cite{tendermint}, HoneyBudger~\cite{Miller:2016}, and Thunderella~\cite{thunderella}, can be viewed as running a permissioned protocol coupled with a higher level \iftoggle{FV}{meta-protocol}{protocol} that continuously selects \iftoggle{FV}{which nodes can}{nodes to} participate in the internal permissioned one.
Variants of this approach are sometimes referred to as \emph{delegated proof of stake (dPoS)} and \emph{Proof of Authority} (PoA).
Hence, any improvement in permissioned protocols will likely yield better unpermissioned  \iftoggle{FV}{protocols as well.}{protocols.}
\nottoggle{VLDB}{
As mentioned above, most solutions in this domain involve repeatedly invoking a consensus instance in order to decide on the next transaction or batch of transactions~\cite{Correia2005,Lamport1998}.
This is vulnerable to performance attacks, as been identified in~\cite{Amir2011}, many of which can be ameliorated by rotating the role of the consensus initial proposer on each invocation of the protocol~\cite{Amir2011,Clement:2009,Veronese:2009}.
}{}

Many of the above works try to optimize performance in the ``common case'' in which there are no failures and the network behaves in a synchronous manner.
These situations are likely to be common in permissioned blockchains, e.g., executed between major financial institutions, established business partners, etc.
Yet, all the above mentioned works run each protocol instance for completion.
Alas, we claim that in a production blockchain system, where transactions are being submitted continuously for as long as the service exists, there is potential for reducing the per transaction and per block communication overhead.
This is by assuming optimistically that the initial proposer of each consensus invocation is correct, and only performing a recovery phase periodically for a batch of affected consensus invocations and only if it is needed.
\nottoggle{FV}{
Our work investigates this approach and develops the \PROTNAME{} protocol to demonstrate its feasibility.
}{}

\iftoggle{SINGLE}
{\subsubsection*{Our Contributions}}
{\paragraph*{\textbf{Our Contributions}}}
We propose \PROTNAME, a new communication frugal optimistic permissioned blockchain protocol.
\PROTNAME{} utilizes the rotating proposer scheme while optimistically assuming that the proposer is correct and that the environment behaves synchronously.
If these assumptions are violated, we do not insist on enforcing agreement immediately.
Instead, we rely on the fact that at least one out of every $f+1$ proposers is correct.
When a correct node discovers, using blockchain's authentication data, that any of the last $f+1$ blocks was not decided correctly in the initial transmission phase, it runs a combined recovery phase for all these incorrectly executed invocations.
This is by invoking a full Byzantine consensus protocol.
At the end of this combined recovery phase, it is ensured that the current prefix of the blockchain is agreed by all correct nodes and will never change as long as there are at most $f$ Byzantine failures.
A single recovery phase may decide the last $f$ blocks, thereby amortizing its cost.

The main benefit of our approach is that when the optimistic assumption holds, the communication overhead of deciding on a block involves a single proposer broadcasting its block and all other nodes broadcasting a single bit of unsigned protocol data.
Further, a new block is being decided in each communication step.
This is by leveraging the iterative nature of blockchain as well as the authentication data that is associated with each block header.

Notice that $f$ is an upper bound on the maximal number of Byzantine nodes in the system.
Yet, in many permissioned blockchains settings, nodes are likely to be highly secured.
Further, in our protocol any Byzantine deviation from the protocol results in a strong proof of which node was the culprit.
Hence, we expect that in real deployments the optimistic assumptions will hold almost always.
In particular, once a proof of Byzantine behavior is being generated, the corresponding Byzantine node will be removed from the system, often resulting in financial penalties and loss of face for the owner of this node.

%To our knowledge, this is the first deterministic consensus algorithm that can achieve agreement in a single %communication step even when there is contention among the nodes~\cite{Guerraoui2010}.

The price paid by our algorithm is that finality of a decision is postponed for $f+1$ invocations (or blocks).
That is, we trade bandwidth and throughput for latency of termination.
As we show when evaluating the performance of our protocol, the average termination latency of blocks is at most a few seconds.
%This latency includes the round-trip latency from the client to the nodes as well as the time until the transaction becomes definitive, i.e., it is included in a block whose depth is at least $f+1$.
In return, when running on non-dedicated virtual machines and network, in a single Amazon data-center, we demonstrate performance of up to $160$K transactions per seconds\iftoggle{VLDB}{ (tps).}{.}
In a non-dedicated multi data-center settings, we obtain up to $30$K \nottoggle{VLDB}{transactions per second.}{tps.}
%In both cases, the throughput is roughly $80\%$ of the network's bandwidth.

%The algorithm's main drawbacks are that \textit{(\romannumeral 1)} the algorithm has to postpone a decision for more $f + 1$ rounds (where $f$ is an upper bound on the number of faulty nodes) and \textit{(\romannumeral 2)} that in a case of failure the performances are rapidly decreasing. Even so, the algorithm's effectiveness is higher in the case of an asynchronous period and benign faults than in the case of Byzantine fault. The fact that blockchain algorithms are meant to a highly reliable environments makes it reasonable to pay with decreased performances once in a while in exchange to extremely high performances most of the time. Respectively, as long the system is under optimal conditions, it is proceeding fast enough such that the latency is quite reasonable. % TODO: is it so even for large f?

	\section{Related Work}
\label{sec:rw}
\iftoggle{VLDB}
{
\paragraph*{Optimistic Consensus}
}
{
\paragraph{Optimistic Consensus}
}
Two main methods were suggested for designing an optimistic consensus protocol:
(i) satisfying safety from the nodes' point of view~\cite{Friedman:2005,Kursawe2002,Martin2006} or (ii) satisfying safety only from the clients' point of view~\cite{Abd-El-Malek:2005,Guerraoui2010,Kotla2007}.
In the first approach, in order to detect inconsistencies, nodes must continuously update other nodes with their state (the exception is~\cite{Friedman:2005} that uses randomization).
In the second approach, nodes are allowed to be temporarily inconsistent with each other.
Only when a client detects an inconsistency, e.g., by receiving inconsistent replies, it initiates a special recovery mechanism to restore the system's consistency.
Concerning blockchains, the first method ignores blockchain's unique features that can be leveraged.
%\begin{comment}
%With the second method, one can either run the blockchain nodes as clients of the agreement service, or run the %consensus algorithm by the blockchain nodes themselves.
%Alas, the first option results in at least two communication steps protocol even in the ``good cases'' while the %second option is not ideal for our purpose as with blockchain Byzantine clients are expected to appear frequently %which may cause the non-optimistic mode to be the common case.	
%\end{comment}
In contrast, running the blockchain nodes as clients of an agreement service results in at least two communication steps protocol even in the ``good cases''.
\iftoggle{VLDB}
{
\paragraph*{Blockchain Systems}
}
{
\paragraph{Blockchain Systems}
}
\nottoggle{VLDB}{To circumvent FLP~\cite{Fischer1985}, most unpermissioned blockchain platforms such as Bitcoin~\cite{nakamoto}, Ethereum~\cite{ethereum} and Algorand~\cite{Gilad:2017} assume a synchronous network, signed messages and a Sybil prevention mechanism.
Bitcoin and the current implementation of Ethereum rely on \emph{Proof-of-Work} (PoW) while Algorand employs \emph{Proof-of-Stake} (PoS). 
In both cases, the algorithm produces a new block in every time slot whose length depends on the worst case  maximal network's latency and block size.
%Thus, PoW and PoS main drawback is that their throughput is bounded by the network's latency.
When the synchrony assumption is violated, safety might be violated as well, resulting in more than one version of the chain that exist concurrently, also known as a \emph{fork} in the chain.
To overcome forks the platform assumes an upper bound on the computation resources that are held by malicious nodes (less than $50$\%) and defines a deterministic rule in which a node always prefers the longest chain it knows about.
Although with PoS the chain may fork with a probability that is less than one in a trillion, with PoW, forks may happen frequently. 
Hence, PoW and the above recovery mechanism impose another crucial drawback: a transaction is never permanent since a longer version may always emerge in the future (although with rapidly diminishing probability). 
Thus, clients cannot rely on a new block until it is deep enough in the chain, resulting in high latency even in the common case.
PoS improves the performance w.r.t. PoW, but results in very complicated protocols.
}
%{
%Loosely speaking, PoW enables scalability to a large number of nodes at the cost of low throughput, significant power consumption, and probabilistic termination.
%PoS improves these aspects, but yields complicated protocols.
%}

Permissioned blockchain protocols typically assume a partially synchronous network while utilizing traditional BFT concepts.
Such platforms run a more computationally efficient protocol than unpermissioned block\-chains but require an a-priori PKI infrastructure.
Traditional BFT solutions are not scalable in the number of participants~\cite{Vukolic15} as their communication complexity grows quadratically in the number of nodes.
%\iftoggle{FV}{
	Hence, such solutions focus on (i) sharding the execution's roles between multiple layers, leaving the consensus to be run by a small set of nodes, and on (ii) designing optimized dedicated BFT consensus protocols. 
	Known platforms like HLF~\cite{bhlfbws, hlfws, abs-1801-10228, abs-1709-06921}
	and R3 Corda \cite{cordaws} offer new models of layered computation and run the BFT-SMaRt~\cite{Bessani2014} protocol,
\nottoggle{VLDB}{
	or a variants of it, as an ordering service layer.
	HLF runs the three phases model of \emph{execute-order-validate}. 
	R3 Corda maintains a hashed directed acyclic graph named Hash-DAG (rather than a single chain), in which a transaction is stored only by those nodes who are affected by it. 
	To ensure transactions' validity, Corda offers a BFT-SMaRt based distributed \emph{notary} service.	
}{or a variants of it, for ordering.}
	
	Platforms such as Chain Core~\cite{ccws,chainWP}, Iroha~\cite{irohaws}, Symb\-oint-Assembly~\cite{symbointws} and Tendermint~\cite{tendermintws} offer new optimized BFT Consensus algorithms.
	Iroha, inspired by the original HLF (v0.6) architecture, runs the Sumeragi consensus protocol which is heavily inspired by BChain~\cite{BChain}.
	BChain is a chain-replication system in which $n$ nodes are linearly arranged and a transaction is moved among the nodes in a chain topology. 
	Namely, each node normally receives a message only from its predecessor. 
	Like \PROTNAME, BChain trades latency for throughput and it has the potential to achieve the best possible throughout~\cite{Guerraoui2010, Guerraoui2010:2}.
	Unlike \PROTNAME, BChain's latency is bounded by at least $n$ rounds. 
	Symboint-Assembly implements its own variant of BFT-SMaRt.
	Tendermint implements an iterative variant of PBFT~\cite{Castro1999} designed by Buchman et el.~\cite{tendermint}. 
	Chain Core runs the Federated consensus protocol in which one node is the leader and $n$ are validators. 
	This protocol is Byzantine resilient for $f < \frac{n}{3}$ only as long the leader is correct.
	Red Belly blockchain \cite{rbbws} offers both, a new computation model that balances the verification load among \emph{verifies} nodes and the Democratic BFT consensus~\cite{Crain2017} that is able to scale the throughput with the number of proposers.
	Finally, HoneyBadger BFT (HBB)~\cite{Miller:2016, hbbws} is a randomized protocol \nottoggle{VLDB}{targeting blockchain. 
	HBB circumvents FLP by randomization and is based on a probabilistic binary Byzantine consensus~\cite{Mostfaoui:2015}.
	}{.}
%}{
%	Hence, such solutions focus on (i) sharding the execution's roles between multiple layers, leaving the consensus to be run by a small set of nodes~\cite{bhlfbws,hlfws,abs-1801-10228,abs-1709-06921}, and on (ii) designing optimized dedicated BFT consensus protocols~\cite{irohaws,symbointws,tendermint,chainWP,Miller:2016} (or both~\cite{rbbws}).
%	Another approach is to rely on randomized protocols, e.g.,~\cite{hbbws,Miller:2016,Mostfaoui:2015}.	
%}

HotStuff~\cite{hotstuff} extends transactions' finality to $3$ rounds and employs signature aggregation~\cite{DLS04} in order to obtain linear communication overhead.
HotStuff requires all nodes to sign an asymmetric signature on each block in the optimistic case while in \PROTNAME{} this is done only by the proposer that generated the block.
Since signing takes pure CPU time, fewer asymmetric signatures enable better throughput.

\nottoggle{VLDB}{
The notion of $k$-coherence was defined in~\cite{ALPT2017}.
\PROTNAME{} can be viewed as guaranteeing an adjustment of $(f+1)$-coherence for deterministic consensus~\cite{Lamport1982, Dwork1988}.
\PROTNAME{} guarantees the classical consensus properties only for decisions at depth greater than $f+1$ and with respect to an external \textsc{valid} method~\cite{CKPS01,Crain2017}.
}{}
	\iftoggle{VLDB}
{
\section{Preliminaries}
\label{sec:prelim}
}
{
\section{Preliminaries and Problem Statement}
\label{sec:prelim}
}

\subsection{System Model}
We consider an asynchronous fully connected environment consisting of $n$ nodes out of which at most $f < \frac{n}{3}$ may incur Byzantine failures~\cite{AW98,Raynal2018book}.
Asynchronous means that no upper bound on the messages' transfer delays exists and nodes have no access to a global clock.
\nottoggle{VLDB}{
While nodes may have access to local clocks, these clocks might not be synchronized with each other and may advance at different rates.}
{}
Fully connected means that any two nodes are connected via a reliable link.
Reliable means that a link does not lose, modify or duplicate a sent message.
Notice that unreliable fair lossy links can be transformed into reliable ones using sequence numbering, retransmissions, and error detection codes~\cite{Coding-Theory-Book}.
A Byzantine failure means that a node might deviate arbitrarily w.r.t. its protocol code including, e.g., sending arbitrary messages, sending messages with different values to different nodes, or failing to send any or all messages.
Yet, we assume that nodes cannot impersonate each other.
A node suffering from a Byzantine failure at any point during its operation is called \emph{Byzantine}; otherwise, it is said to be \emph{correct}.

\nottoggle{VLDB}{
Following the FLP result~\cite{Fischer1985}, consensus is unsolvable in a truly asynchronous system.
We circumvent FLP by enriching the system with the $\bDiamond \mathit{Synch}$  assumption~\cite{Bouzid2015}.
$\bDiamond \mathit{Synch}$ means that after an unknown time $\tau$ there is an unknown upper bound $\delta$ on a message's transfer delay.
}
{In order to circumvent the FLP result~\cite{Fischer1985}, we enrich the system with the $\bDiamond \mathit{Synch}$  assumption~\cite{Bouzid2015}.
That is, after an unknown time $\tau$ there is an unknown upper bound $\delta$ on a message's transfer delay.
}
As in most Byzantine fault tolerance works~\cite{Guerraoui2010,Castro1999,Kotla2007}, $\bDiamond \mathit{Synch}$  is only needed to ensure liveness, meaning that even under severe network delays safety is never violated.
Finally, a node may sign a message by an unforgeable signature.
We denote the signature of node $p$ on message $m$ by $sig_p(m)$.
The implementation of the signature mechanism is done by a well known cryptographic technique, such as symmetric ciphers~\cite{Ebrahim2014}, RSA~\cite{Rivest1978} or an elliptic curves digital signature (ECDS)~\cite{Johnson2001}.
	
\subsection{Underlying Protocols}
Solutions to the following fundamental distributed computing problems serve as building blocks for \PROTNAME.

\paragraph*{Reliable Broadcast (RB)}
The reliable broadcast abstraction~\cite{Bracha1987} (denoted \textbf{RB-Broadcast}) ensures reliable message delivery in the presence of Byzantine failures.
To utilize \textbf{RB-Broadcast} nodes may invoke two methods: \emph{RB-broadcast} and \emph{RB-deliver}.
A correct node that wishes to broadcast a message $m$ invokes \emph{RB-broadcast(m)} while a node that expects to receive a message invokes \emph{RB-deliver}.
By a slight abuse of notation, we denote \emph{RB-deliver(m)} the fact than an invocation of \emph{RB-deliver} returned the message $m$ and say that the invoking process has RB-delivered $m$.
\iftoggle{VLDB}
{The \textbf{RB-Broadcast} abstraction satisfies the following:}
{The \textbf{RB-Broadcast} abstraction satisfies the following properties:}
\iftoggle{FV}{
\iftoggle{MYACM}{
	\begin{description}
}{	
	\begin{LaTeXdescription}
}
		\item[RB-Validity:] If a correct node has RB-delivered a message $m$ from a correct node $p$, then $p$ has invoked \emph{RB-broadcast(m)}.
		%\item \emph{RB-Unicity:} For a given unique message $m$, every correct node \emph{RB-delivers} $m$ at most once.
		\item[RB-Agreement:] If a correct node \emph{RB-delivers} a message $m$, then all correct nodes eventually \emph{RB-deliver} $m$.
		\item[RB-Termination:] If a correct node invokes \emph{RB-broadcast(m)}, then all correct nodes eventually \emph{RB-deliver} $m$.
	\iftoggle{MYACM}{
	\end{description}
	}{	
	\end{LaTeXdescription}
}
}{\\
	\noindent\textbf{*~RB-Validity:} If a correct node has RB-delivered a message $m$ from a correct node $p$, then $p$ has invoked \emph{RB-broadcast(m)}.\\
	%\item \emph{RB-Unicity:} For a given unique message $m$, every correct node \emph{RB-delivers} $m$ at most once.
	\noindent\textbf{*~RB-Agreement:} If a correct node \emph{RB-delivers} a message $m$, then all correct nodes eventually \emph{RB-deliver} $m$.\\
	\noindent\textbf{*~RB-Termination:} If a correct node invokes \emph{RB-broadcast(m)}, then all correct nodes eventually \emph{RB-deliver} $m$.
}

\paragraph*{Atomic Broadcast (AB)}
\iftoggle{FV}{
The atomic broadcast abstraction~\cite{Cristian95} is more restrictive than \textbf{RB-Broadcast} in the sense that in addition to the \emph{RB-Broadcast} properties it requires the \emph{Order} property as well.
}
{
	Atomic broadcast~\cite{Cristian95} adds the following \emph{Order} property to \textbf{RB-Broadcast}.
}
\iftoggle{FV}{
\iftoggle{MYACM}{
	\begin{description}
	}{	
	\begin{LaTeXdescription}
}
		\item[Atomic-Order:] All messages delivered by correct nodes are delivered in the same order by all correct nodes.
	\iftoggle{MYACM}{
	\end{description}
}{	
	\end{LaTeXdescription}
}
}{
\\
\noindent\textbf{*~Atomic-Order:} Messages delivered by correct nodes are delivered in the same order by all correct nodes.
}

\nottoggle{VLDB}{
\paragraph*{Multi-value Byzantine Consensus}
Most implementations of atomic broadcast rely on consensus protocols as sub-routines.
In our context, \emph{Multi-value Byzantine Consensus} (MVC)~\cite{Lamport1982} is a variant of the consensus problem in which a set of nodes, each potentially proposing its own value, must decide the same value in a decentralized network despite the presence of Byzantine failures.
A solution to the MVC problem satisfies the following properties~\cite{Dwork1988}:
	%source
\iftoggle{MYACM}{
	\begin{description}
	}{	
		\begin{LaTeXdescription}
		}
		\item[MVC-Validity:] If all correct nodes have proposed the same value $v$, then $v$ must be decided.
		\item[MVC-Agreement:] No two correct nodes decide differently.
		\item[MVC-Termination:] Each correct node eventually decides.
	\iftoggle{MYACM}{
\end{description}
}{	
\end{LaTeXdescription}
}
Notice that according to \textbf{MVC-Validity}, any decision is valid when not all correct nodes propose the same value.
Hence, it is not helpful when the system begins in a non agreement state, which is common in real systems.
To that end, an extended validity property was suggested in~\cite{Correia2005} and is composed of three sub-properties:
\iftoggle{MYACM}{
	\begin{description}
	}{	
		\begin{LaTeXdescription}
		}
		\item[MVC1-Validity:] If all correct nodes have proposed the same value $v$, then $v$ must be decided.
		\item[MVC2-Validity:] A decided value $v$ was proposed by some node or $v=nil$.
		\item[MVC3-Validity:] No correct node decides a value that was proposed by only faulty nodes.
	\iftoggle{MYACM}{
\end{description}
}{	
\end{LaTeXdescription}
}
This validity property defines that if no correct node has suggested $v$, then $v$ cannot be decided.
This definition states precisely what is the set of valid values in any possible run of the algorithm.

\paragraph*{(Optimistic) Binary Byzantine Consensus}
The \emph{Binary Byzantine Consensus} (BBC) is the simplest variant of MVC in which only two values are possible.
The \textbf{MVC-Validity}, \textbf{MVC-Agreement}, and \textbf{MVC-Termination} properties naturally translate to their corresponding \textbf{BBC-Validity}, \textbf{BBC-Agreement}, and \textbf{BBC-Termination} properties.
Interestingly, BBC cannot be solved in a weaker model than MVC.
However, as only two values are possible, an \emph{Optimistic BBC} (OBBC) is capable of achieving an agreement in a single communication step if a predefined set of conditions is met \cite{BF2020,Friedman:2005, Mosfaoui:2015}.
}{
\paragraph*{(Optimistic) Binary Byzantine Consensus}
\emph{Binary Byzantine Consensus} (BBC) is the simplest variant of \emph{Multi-value Byzantine Consensus} (MVC)~\cite{Lamport1982} in which only two values are possible. 
A solution to the BBC problem satisfies the following properties~\cite{Dwork1988, Correia2005}:
\begin{description}
\item[BBC-Validity:] If all correct nodes have proposed the same value $v$, then $v$ must be decided.
\item[BBC-Agreement:] No two correct nodes decide differently.
\item[BBC-Termination:] Each correct node eventually decides.
\end{description}
%\begin{LaTeXdescription}
	%\item[BBC-Validity:] If all correct nodes have proposed the same value $v$, then $v$ must be decided.
	%\item[BBC-Agreement:] No two correct nodes decide differently.
	%\item[BBC-Termination:] Each correct node eventually decides.
%\end{LaTeXdescription}

As there are only two possible values in BBC, an \emph{Optimistic BBC} (OBBC) is capable of achieving an agreement in a single communication step if a predefined set of favorable conditions are met~\cite{BF2019full,BF2020,Friedman:2005,Mosfaoui:2015}.
}
 
\subsection{Problem Statement}
Blockchain algorithms require an external validity mechanism as sometimes even Byzantine nodes may propose legal values (or blocks)~\cite{Luu:2016}.
Therefore, the validity of a value may be defined by an external predefined method.
The \emph{Validity Predicate-based Byzantine Consensus} (VPBC)~\cite{CKPS01,Crain2017} abstraction captures this observation by defining the following validity property:
\iftoggle{FV}{
\iftoggle{MYACM}{
	\begin{description}
	}{	
		\begin{LaTeXdescription}
		}
	\item[VPBC-Validity:] A decided value satisfies an external predefined \textsc{valid} method.
	\iftoggle{MYACM}{
\end{description}
}{	
\end{LaTeXdescription}
}
}{
\\
\noindent\textbf{*~VPBC-Validity:} A decided value satisfies an external predefined \textsc{valid} method.\\
}
\textbf{VPBC-Agreement} and \textbf{VPBC-Termination} are the same as their MVC/BBC counterparts.
That is,  VPBC  generalizes the classical definition of MVC.

%\paragraph*{Iterative Validity Predicate-based Consensus}
In a blockchain, each block carries a glimpse to its creator knowledge of the system's state.
This glimpse is encapsulated in the hash that each block carries.
In order to leverage the iterative nature of blockchains, we define a weaker model than \emph{VPBC} in which we denote each iteration with a round number $r$.
Next, we define the following per round notions:
\iftoggle{MYACM}{
	\begin{description}
	}{	
		\begin{LaTeXdescription}
		}
	\item[Tentative decision:] A decision of the protocol at a given node and round that might still be changed.
	\item[Definite decision:] A decision of the protocol at a given node and round that will never change.
	\item[\boldmath{$v_{p}^{r}$}:] A value that was decided or received by $p$ in round $r$ of the protocol. $v^r$ denotes a value that was decided or received by some node in round $r$.
	\item[\boldmath{$d(v_p^r)$}:] Let $r'$ be the current round of the protocol that node $p$ runs. For a given $v_{p}^{r}$ (possible tentative), we denote its \emph{depth} as $d(v_{p}^{r}) = r' - r$.
	\iftoggle{MYACM}{
\end{description}
}{	
\end{LaTeXdescription}
}
\begin{comment}
\begin{table}
\centering
\begin{tabu} to 1 \textwidth {|| m{2cm} | m{5cm} ||}
\hline
\emph{Tentative decision} & \makecell{A decision of the protocol\\at a given node and round that might still be changed}\\
\hline \emph{Definite decision} & \makecell{A decision of the protocol\\at a given node and round that will never be changed}\\
\hline $v_{p}^{r}$ & \makecell{A value that was decided or received by $p$ in round $r$ of the protocol.\\$v^r$ denotes a value that was decided or received by some node in round $r$} \\
\hline
$d(v_p^r)$ & \makecell{Let $r'$ be the current round of the protocol that node $p$ runs.\\For a given $v_{p}^{r}$ (possible tentative), we denote its \emph{depth} as $d(v_{p}^{r}) = r' - r$}\\
\hline
\end{tabu}
\end{table}
\end{comment}

\iftoggle{VLDB}
{\begin{definition}[\FULLDEF{}] (\PREFDEFF)}
{\begin{definition}[\FULLDEF{} (\PREFDEF{})]}
Let \textsc{valid} be a predefined method as in VPBC and let $\rho$ be a predefined fixed constant.
The $\rho$-\FULLDEF{} (\ACRDEF{$\rho$}) abstraction defines the following properties:
\iftoggle{FV}{
\iftoggle{MYACM}{
	\begin{description}
	}{	
		\begin{LaTeXdescription}
		}
	\item[\PREFDEF{}-Validity:] A decided (possible tentative) value $v$ satisfies the \textsc{valid} method.
	\item[\PREFDEF{}-Agreement:] For any two correct nodes $p,q$, let $v_p^r, v_q^r$ be their decided value in round $r$. If $d(v_p^r), d(v_q^r) > \rho$ then $v_p^r=v_q^r$.
	\item[\PREFDEF{}-Termination:] Every round eventually terminates.
	\item[\PREFDEF{}-Finality:] In every round $r' > r + \rho$, $v_p^r$ is definite.
	\iftoggle{MYACM}{
\end{description}
}{	
\end{LaTeXdescription}
}
}{
\\
\noindent\textbf{*~\PREFDEF{}-Validity:} A decided (possible tentative) value $v$ satisfies the \textsc{valid} method.\\
\noindent\textbf{*~\PREFDEF{}-Agreement:} For any two correct nodes $p,q$, let $v_p^r, v_q^r$ be their decided value in round $r$. If $d(v_p^r), d(v_q^r) > \rho$ then $v_p^r=v_q^r$.\\
\noindent\textbf{*~\PREFDEF{}-Termination:} Every round eventually terminates.\\
\noindent\textbf{*~\PREFDEF{}-Finality:} In every round $r' > r + \rho$, $v_p^r$ is definite.\\
}
\end{definition}
\PREFDEF{} guarantees the VPBC's properties only for decisions at depth greater than $\rho$.
With blockchains, as every block contains an authentication data regarding its predecessors, it provides a cryptographic summary of its creator history.
This \iftoggle{VLDB}{information assists}{information may assist} in detecting failures without the necessity of sending more information.
In addition, blocks are continuously added to the chain.
Thus, a block eventually becomes deep enough such that it satisfies the standard VPBC properties.
Let us note that the \PREFDEF{}-Agreement property is similar to the notion of \emph{common prefix} in~\cite{Garay15}.

In a blockchain setting, \emph{clients} of the system submit \emph{transactions} to the nodes, and the decisions values are blocks, each consisting of zero or more transactions previously submitted by clients.
In case the \textsc{valid} method may accept empty blocks, we would like to prevent trivial implementations in which every node locally generates empty blocks continuously.
Obviously, the throughput of such a protocol would be $0$, and thus it would be considered useless.
Yet, in order to prove that a protocol does not unintentionally suffer from such a behavior even under Byzantine failures we add the following requirement:
\iftoggle{FV}{
\iftoggle{MYACM}{
	\begin{description}
	}{	
		\begin{LaTeXdescription}
		}
	\item[Non-Triviality:] If clients repeatedly submit transactions to the system, then the nodes repeatedly decide definitively non-empty blocks. 
	\iftoggle{MYACM}{
\end{description}
}{	
\end{LaTeXdescription}
}
}{
\\
\noindent\textbf{*~Non-Triviality:} If clients repeatedly submit transactions to the system, then the nodes repeatedly decide definitively non-empty blocks. 
}
		\section{Weak Reliable Broadcast}
\label{sec:wrb}

\subsection{Overview}
\label{sec:wrb-overview}
The \emph{Weak Reliable Broadcast} (WRB) abstraction serves as \PROTNAME's main message dissemination mechanism.
WRB offers weaker agreement guarantees than \textbf{RB-Broadcast}~\cite{Bracha1987}.
In general, WRB ensures that the nodes agree on ($i)$ the sender's identity and ($ii$) whether to deliver a message at all, rather than the content of the message\footnote{
	To the best of our knowledge, we are the first to discuss WRB.
	Bracha's approach~\cite{Bracha1987} can be viewed as the opposite, first agree on the content of the message (\emph{consistent broadcast}), and then agree whether it should be delivered.}.
WRB is associated with the \emph{WRB-broadcast} and \emph{WRB-deliver} methods.
A node that wishes to disseminate a message $m$ invokes \emph{WRB-broadcast(m)}.
If a node expects to receive a message $m$ from $k$ through this mechanism it invokes \emph{WRB-deliver(k)}.
\emph{WRB-deliver(k)} returns a message $m$ where $m$ is the received message.
If the nodes were not able to deliver $k$'s message, \emph{WRB-deliver(k)} returns \emph{nil}.
Formally, WRB satisfies the following properties:
\iftoggle{FV}{
\iftoggle{MYACM}{
	\begin{description}
	}{	
		\begin{LaTeXdescription}
		}
	\item[WRB-Validity:] If a correct node \emph{WRB-delivered(k)} $m \neq nil$, then $k$ has invoked \emph{WRB-broadcast(m)}.
	\item[WRB-Agreement:] If two correct nodes $p,q$ \emph{WRB-delivered(k)} $m_p, m_q$ respectively from $k$, then either $m_p = m_q = nil$ or $m_p \neq nil \land m_q \neq nil$. 
	%Namely, all correct nodes either deliver a message different than $nil$ from $k$ or they all deliver $nil$.
	\item[WRB-Termination:] If a correct node $p$ \emph{WRB-deliver(k)} $m$ from $k$, then every correct node that is trying to \emph{WRB-deliver(k)} eventually \emph{WRB-deliver(k)} some message $m'$ from $k$.
	\item[WRB-Non-Triviality:] If a correct node $k$ repeatedly invokes \emph{WRB-broadcast(m)} then eventually all correct nodes will \emph{WRB-deliver(k)} $m$.
	\iftoggle{MYACM}{
\end{description}
}{	
\end{LaTeXdescription}
}
}{
\\
\noindent\textbf{*~WRB-Validity:} If a correct node \emph{WRB-delivered(k)} $m \neq nil$, then $k$ has invoked \emph{WRB-broadcast(m)}.\\
\noindent\textbf{*~WRB-Agreement:} If two correct nodes $p,q$ \emph{WRB-delivered(k)} $m_p, m_q$ respectively from $k$, then either $m_p = m_q = nil$ or $m_p \neq nil \land m_q \neq nil$.\\ %Namely, all correct nodes either deliver a message different than $nil$ from $k$ or they all deliver $nil$.\\
\noindent\textbf{*~WRB-Termination:} If a correct node $p$ \emph{WRB-deliver(k)} $m$ from $k$, then every correct node that is trying to \emph{WRB-deliver(k)} eventually \emph{WRB-deliver(k)} some message $m'$ from $k$.\\
\noindent\textbf{*~WRB-Non-Triviality:} If a correct node $k$ repeatedly invokes \emph{WRB-broadcast(m)} then eventually all correct nodes will \emph{WRB-deliver(k)} $m$.
}

\subsection{Implementing WRB}
\label{sec:implWRB}

\iftoggle{VLDB}{\footnotesize}{}
\begin{algorithm}[t]
	\caption{Weak Reliable Broadcast - code for $p$}
    \label{alg:wrb}
    \setcounter{AlgoLine}{0}
    \SetNlSty{texttt}{[}{]}
    \footnotesize
    \SetKwFunction{BR}{broadcast}
    \SetKwFunction{Send}{send}
    \SetKwFunction{WRBB}{WRB-broadcast}
    \SetKwFunction{WRBD}{WRB-deliver}
    \SetKwRepeat{Wait}{wait}{until}
    \SetKwFor{Upon}{upon}{do}{end} 
    \SetKwProg{myproc}{Procedure}{}{}
%	\begin{algorithmic}[1]
		$\mathit{timer} \leftarrow \tau$\; \label{wrb:l1} 
	\myproc{\WRBB{$m$}}{
		
		\BR($m, \mathit{sig_p}(m)$); \label{wrb:l2} \tcc{push phase} 
	}
	\BlankLine
	\myproc{\WRBD{$k$}}{
		start $\mathit{timer}$\; \tcc{$\mathit{timer}$'s value is set in lines \ref{wrb:l1}, \ref{wrb:l6} and \ref{wrb:l9} of WRB-deliver's last invocation}
		\lWait{a valid $(m, \mathit{sig_k}(m))$ has been received or $\mathit{timer}$ has expired
		\label{wrb:l3}
		}{} 
		
		\eIf{a valid $(m, \mathit{sig_k}(m))$ has been received \label{wrb:l4}
		} 
		{
			$d \leftarrow \mathit{OBBC.propose}(1)$\;
			\tcc{If no node has proposed 0, then $\mathit{OBBC.propose}$ ends in a single communication step}
			
		}{
			$d \leftarrow \mathit{OBBC.propose(0)}$\;
		} \label{wrb:l5}
		increase $\mathit{timer}$\; \label{wrb:l6}
		\If{$d = 0$} {
			return $\mathit{nil}$\;
		} \label{wrb:l7}
		\If{a valid $(m, \mathit{sig_k}(m))$ has been received \label{wrb:l8}}{ 
			adjust $\mathit{timer}$\; \label{wrb:l9}
			return $m$\;
		}	\label{wrb:l10}
		\BR{$\mathit{REQ}, k$}\; \label{wrb:l11}
		\lWait{a valid $(m', \mathit{sig_k}(m'))$ has been received}{}
		 return $m'$\; \label{wrb:l12}
		 \BlankLine
		 \Upon{receiving $(\mathit{REQ}, k)$ from $q \land$ a valid $(m, \mathit{sig_k}(m))$ has been received \label{wrb:l13}}{ 
	 		\Send{$m$, $\mathit{sig_k}(m)$)} to $q$\;
		 } \label{wrb:l14}
	}		
\end{algorithm}
\iftoggle{VLDB}{\normalsize}{}

The pseudo-code implementation of WRB is listed in Algorithm~\ref{alg:wrb}.
To \emph{WRB-broadcast} a message, a node simply broadcasts it to everyone (line~\ref{wrb:l2}).
When $p$ invokes \emph{WRB-deliver(k)} it performs the following:
\begin{itemize}
	\item It waits for at most $timer$ to receive $k$'s message (line~\ref{wrb:l3}).
	\item If such a valid message has been received, then $p$ votes to deliver it using an OBBC protocol.
	Else, $p$ votes against delivering $k$'s message (lines~\ref{wrb:l4}-\ref{wrb:l5}).
	Recall that if no node has proposed 0, then OBBC ends in a single communication step.
	\item If the decision is not to deliver (OBBC returned 0), $p$ returns $\mathit{nil}$ and increases the timer (lines~\ref{wrb:l6}-\ref{wrb:l7}).
	\item If it is decided to deliver the message (OBBC returned 1) and the message has already been received by $p$, then $p$ adjusts the timer and returns $m$ (lines~\ref{wrb:l8}-\ref{wrb:l10}).
	\item Else, OBBC decided $1$, meaning that at least one correct node received $k$'s message and voted for its acceptance. Thus, $p$ moves to a \emph{pull} phase and pulls $k$'s message from the nodes who did receive it.
	When $p$ eventually receives a valid message, $p$ returns it (lines~\ref{wrb:l11}-~\ref{wrb:l12}).
	\item Upon receiving $q$'s request for $k$'s message, if $p$ has $k$'s message $m$, it sends $(m, \mathit{sig_k}(m))$ \nottoggle{VLDB}{back}{} to $q$ (lines~\ref{wrb:l13}-\ref{wrb:l14}).
	
%	\item On \iftoggle{FV}{obtaining}{} a valid response to $p$'s republish request, $p$ returns the received message (lines~\ref{alg:wrb:deliver-wait-res}-\ref{alg:wrb:deliver-return-res}).	 
\end{itemize}

To ensure liveness, the \emph{timer} is increased each time $p$ does not receive the message  (line~\ref{wrb:l6}).
To avoid having too long timers for too long, \emph{timer} is adjusted downward when a message is received by $p$ (line~\ref{wrb:l9}).
The exact details \nottoggle{VLDB}{of such adjustments}{} are beyond the scope of this paper, but see for example~\cite{Castro1999}.

\nottoggle{VLDB}{
In a typical implementation of OBBC each node broadcasts its vote~\cite{BF2020,Raynal2018book}. % {\color{red} [which paper?]}~\cite{}.
Then if a node receives enough votes for the same value $v$ to safely decide $v$ after this single communication step, it decides $v$ and returns.
We present our own OBBC protocol in Appendix~\ref{app:obbc}.
}
{
In a typical implementation of OBBC each node broadcasts its vote~\cite{BF2020,Raynal2018book}. % {\color{red} [which paper?]}~\cite{}.
Then if a node receives enough votes for the same value $v$ to safely decide $v$ after this single communication step, it decides $v$ and returns.
We present our own OBBC protocol in the full version of this paper~\cite{BF2019full}.
}
%[[{\color{red} considering the extra explanations on OBBC in the algorithm pseudo code and description, can we remove the above?}]]

%\paragraph*{Piggybacking Over WRB-Deliver}
%In order for \PROTNAME{} to maintain its amortized one communication step per block we need \PROTNAME{} to be able to piggyback one of its %message types on another message it sends as part of the \emph{WRB-delivers} protocol.
%The details appear in Section~\ref{sec:desc} below.
%We support this by augmenting the \emph{WRB-deliver} method to receive two parameters, \emph{WRB-deliver(k,pgd)}, such that \emph{pgd} is the potential piggybacked data (can be $\mathit{nil}$).
%The OBBC protocol is also augmented to receive \emph{pgd} and piggyback it on the first message it broadcasts.
%Recall that we assume an OBBC protocol that always starts by having each node broadcast its vote.
%Hence in the augmented protocol, each node that starts the OBBC protocol broadcasts its vote alongside the piggybacked \emph{pgd} message, %which is made available to the calling code together with the decision value.
\nottoggle{VLDB}{
\subsection{WRB Correctness Proof}
\begin{lemma}[\emph{WRB-Validity}]
\label{lemma:wrb-validity}
	If a correct node $p$ \emph{WRB-delivered(k)} a message $m \neq nil$, then $k$ has invoked \emph{WRB-broadcast(m)}.
\end{lemma}

\begin{proof}
The system model assumes reliable channels.
Also, each node signs the messages it \emph{WRB-broadcasts}.
Hence, if a message $m$ is received and pass validation as being sent by $k$, then $k$ has indeed \emph{WRB-broadcast(m)}.
\end{proof}

\begin{lemma}[\emph{WRB-Agreement}]
\label{lemma:wrb-agreement}
If two correct nodes $p,q$ \emph{WRB-delivered(k)} $m_p, m_q$ respectively from $k$, then $m_p,m_q \neq nil$.
Namely, all correct nodes either deliver a message different than $nil$ from $k$ or they all deliver $nil$.
\end{lemma}

\begin{proof}
Let $p,q$ be two correct nodes that have \emph{WRB-delivered(k)} $m_p, m_q$ respectively from $k$.
Assume b.w.o.c and w.l.o.g that $m_p \neq nil$ while $m_q = nil$.
A correct node may return $nil$ only if \emph{OBBC.propose} returns $0$.
Respectively, a correct node returns $m\neq nil$ only if \emph{OBBC.propose} returned $1$.
This is in contradiction to \textbf{BBC-Agreement}.
Hence, either $m_p, m_q \neq nil$ or both $p$ and $q$ return $nil$.
\end{proof}

\begin{lemma}[\emph{WRB-Termination}]
\label{lemma:wrb-termination1}
If a correct node $p$ \emph{WRB-deliver(k)} $m$ from $k$, then every correct node that is trying to \emph{WRB-deliver(k)} eventually \emph{WRB-deliver(k)} some message $m'$ from $k$.
\end{lemma}

\begin{proof}
A node $p$ waits for $k$'s message at most $timer$ time units.
Hence $p$ eventually invokes \emph{OBBC.propose}.
By \textbf{BBC-Termination}, no correct node is blocked forever on \emph{OBBC.propose}.
If \emph{OBBC.propose} has returned $0$ then $p$ returns $nil$.
Else, if $k$'s message has been received by $p$ it returns $m$.
Otherwise, $p$ asks from other nodes for $k$'s message.
By \textbf{BBC-Validity}, if \emph{OBBC.propose} returns a value $v$ then at least one correct node $q$ proposed $v$.
Thus, if \emph{OBBC.propose} returned $1$ it means that at least one correct node $q$ has proposed $1$ and by the algorithm's code $q$ does so if it has received $k$'s message $m$.
Hence, when $q$ receives $p$'s request it sends back $m$.
As $q$ is correct $p$ receives and returns~$m$.
\end{proof}

\begin{lemma}[\emph{WRB-Non-Triviality}]
\label{lemma:wrb-termination2}
If a correct node $k$ repeatedly invokes \emph{WRB-broadcast(m)} then eventually all correct nodes will \emph{WRB-deliver(k)} $m$.
\end{lemma}

\begin{proof}
By the algorithm's code, on any unsuccessful delivery the timer increases.
Let $k$ be a correct node that repeatedly invokes \emph{WRB-broadcast(m)}.
Following the $\bDiamond Synch$ \cite{Bouzid2015} assumption and the incremental timer method \cite{Castro1999}, eventually after an unknown time $\tau$ the timer of all nodes is longer than the unknown upper bound $\delta$ on a message's transfer delay.
Hence, every correct node receives $k$'s message and invokes \emph{OBBC.propose(1)}.
By \textbf{BBC-Agreement} if all correct nodes propose $1$ then $1$ has to be decided.
Hence, every correct node eventually \emph{WRB-deliver(k)} $m$.
\end{proof}

By the above four lemmas we have:
\begin{theorem}
The protocol listed in Algorithm~\ref{alg:wrb} solves WRB.
\end{theorem}
Note that WRB's communication costs depend on the success of OBBC. 
If OBBC decides in a single communication step, then WRB-deliver also terminates in a single communication step.
%\begin{observation}
%	\label{wrb:obs1}
%	If for every node OBBC decides in a single communication step, then WRB-deliver also terminates in a single communication step.
%\end{observation}
}
{
Since the correctness proof of Algorithm~\ref{alg:wrb} is technical, it is deferred to the full version of this paper~\cite{BF2019full}.
}

	\iftoggle{VLDB}{
\section{The \PROTNAME{} Protocol}
}
{
\iftoggle{SINGLE}{
\section{The \PROTNAME{} Protocol}
}{
\section{\PROTNAME{} -- a Total ordering Optimistic sYstem}
}
}
\label{sec:toy}

\nottoggle{VLDB}{
\PROTNAME{} implements the \ACRDEF{$\rho$} abstraction with $\rho = f + 1$.
This section focuses on the algorithmic aspects of \PROTNAME{} and its correctness, while actual implementation considerations and performance are discussed in Section~\ref{sec:impl} below.
}{
}

We present \PROTNAME{} in a didactic way:
We first show a two-phased \emph{crash fault tolerant} (CFT) ordering protocol based on WRB.
We then improve it to a single-phased protocol.
Finally, we extend the protocol to tolerate Byzantine failures.
The pseudo-code of the protocol appears in Algorithm~\ref{cft-bft} and~\ref{recovery}.
Regularly numbered lines correspond to the CFT aspects of the protocol, while lines prefixed with `b' are the additions to handle Byzantine failures.

\iftoggle{VLDB}{
	\begin{algorithm}[t!]
}
{
	\begin{algorithm}[t]
}
	\footnotesize
	\SetAlgoLined
	\SetKwFunction{WRBD}{WRB-deliver}
	\SetKwFunction{WRBB}{WRB-broadcast}
	\SetKwFunction{RBB}{RB-broadcast}
	\SetKwFunction{RBD}{RB-deliver}
	\SetKwFunction{ABB}{Atomic-broadcast}
	\SetKwFunction{ABD}{Atomic-deliver}
	\SetKwFunction{REC}{RECOVERY}
 	\SetKwFor{Upon}{upon}{do}{end} 
%	\begin{algorithmic}[1]
		$r_i \leftarrow 0$\;
		$\mathit{proposer_{r_i}}\leftarrow p_0$\;
		$\mathit{full\_mode \leftarrow true}$\;
		\While{true}{
		$b \leftarrow nil$\;
		\setcounter{AlgoLine}{0}
		\SetNlSty{texttt}{[b}{]}
		\While{$\mathit{proposer_{r_i}}$'s block was tentatively decided in the last $f$ rounds \label{bft:l1}}   
		{ 
		$\mathit{proposer_{r_i}} \leftarrow (\mathit{proposer_{r_i}} + 1) \mod n$\;
		}  \label{bft:l2}
		\setcounter{AlgoLine}{5}
		\SetNlSty{texttt}{}{}
		\If {$i=\mathit{proposer_{r_i}} \land \mathit{full\_mode = true}$ \label{cft:l1} }  { \tcc{executed if nil has been WRB-delivered in the last iteration} 
		$b \leftarrow$ prepared block\;
		\WRBB{$b$}\;
		$ b \leftarrow nil$\;
		} \label{cft:l2}

		\If{$(\mathit{proposer_{r_i}} + 1) \mod{n} = i$ \label{cft:l3}} { 
		$b \leftarrow$ prepared block\;
		} \label{cft:l4}
		$b_i^{r_i} \leftarrow$ \WRBD{$\mathit{proposer_{r_i}}, b$}\; \label{cft:l5}
		\If {$b_i^{r_i} = nil$} { \label{cft:20}
		 $\mathit{full\_mode \leftarrow true}$\;
		$\mathit{proposer_{r_i}} \leftarrow (\mathit{proposer_{r_i}} + 1) \mod n$\;
		continue\;
		} \label{cft:l7}
		$\mathit{full\_mode \leftarrow false}$\;
		\setcounter{AlgoLine}{3}
		\SetNlSty{texttt}{[b}{]}
		\If {$b_i^{r_i} $ is not $valid$ \label{bft:l3}} {
		\tcc{validating the block hash against the previous block} 
		$proof\leftarrow (b_i^{r_i}, sig_{\mathit{proposer_{r_i}}}(b_i^{r_i}),
			\newline b_i^{r_i-1}, sig_{\mathit{proposer_{r_i - 1}}}(b_i^{r_i-1}))$\;
		\RBB{$proof$}\;
		invoke \REC{$r_i$, $proof$}\; 
		continue\;
		}\label{bft:l4}
		\setcounter{AlgoLine}{21}
		\SetNlSty{texttt}{}{}
		append $b_i^{r_i}$ to the chain\label{cft:l6}\;
		\setcounter{AlgoLine}{10}
		\SetNlSty{texttt}{[b}{]}
		decide $b_i^{r_i - (f + 2)}$\label{bft:l12}\;
		\setcounter{AlgoLine}{22}
		\SetNlSty{texttt}{}{}
		$\mathit{proposer_{r_i}} \leftarrow (\mathit{proposer_{r_i}} + 1) \mod n$\; \label{cft:l8}
		$r_i \leftarrow r_i + 1$\; 
		} \label{cft:l9}
		\BlankLine
		\setcounter{AlgoLine}{11}
		\SetNlSty{texttt}{[b}{]}
		\Upon{\RBD a valid $proof \leftarrow (b^{r}, sig_{\mathit{proposer_{r}}}(b^{r}), b^{r-1}, sig_{\mathit{proposer_{r - 1}}}(b^{r-1}))$ \label{bft:l5}}
		{
				invoke \REC{$r$, $proof$}\; 
		} \label{bft:l6}
	
	\caption{\PROTNAME{} -- code for $p_i \newline$ \footnotesize The lines that start with `b' depict the BFT additions}
	\label{cft-bft}
\end{algorithm}

\iftoggle{VLDB}{}{
\begin{algorithm}[t]
	%	\caption{\PROTNAME{} -- code for $p_i$\\ \footnotesize The black and red lines depict the CFT and BFT mechanisms respectively}
	%	\label{cft-bft}
	\footnotesize
	\SetAlgoLined
	\SetKwFunction{ABB}{Atomic-broadcast}
	\SetKwFunction{ABD}{Atomic-deliver}
	\SetKwFunction{REC}{RECOVERY}
	
	\setcounter{AlgoLine}{0}
	\SetNlSty{texttt}{[}{]}
	\SetKwProg{myproc}{Procedure}{}{}
	\myproc{\REC{$r$, $proof$}}{
		$\mathit{versions_r}\leftarrow \{\}$\label{bft:l8}\;
		\eIf{$r_i < r - 1$}
		{
			$v \leftarrow \mathit{empty\_version}$\;
		}{
			$v \leftarrow [(b^{r - (f + 1)}, sig_{\mathit{proposer_{r - (f + 1)}}}(b^{r - (f + 1)})), ...,
			\newline (b^{r - 1}, sig_{\mathit{proposer_{r - 1}}}(b^{r - 1})), ..., \newline (b^{r_i}, sig_{\mathit{proposer_{r_i}}}(b^{r_i}))]$\;
		}
		\ABB{$v$}\;  \label{bft:l9}
		\Repeat{$|\mathit{versions_r}| = n - f$}
		{\label{bft:l10}
			\ABD $v_j$ from $p_j$\; 
			\If{$v_j$ is $valid$} {
				\tcc{validating the received version} 
				$\mathit{versions_r} \leftarrow \mathit{versions_r} \cup \{v_j\}$\;
			} 
			
		}
		$v' \leftarrow $ the first received among $\{v_j \in \mathit{versions_r} |r_j = max\{r_k | v_k \in \mathit{versions_r} \land (b^{r_k}, sig_{\mathit{proposer_{r_k}}}(b^{r_k}) \in v_k\} \}$\;
		adopt $v'$ and update $r_i$, and $\mathit{proposer_{r_i}}$\;
		$\mathit{full\_mode \leftarrow true}$;
	} \label{bft:l11}
	\caption{Recovery Procedure -- code for $p_i$}
	\label{recovery}
\end{algorithm}
}

\subsection{Simplified CFT \PROTNAME{}}
As mentioned in Section~\ref{sec:wrb}, WRB supports an all or nothing delivery that is blind to the message's content. 
When there are no Byzantine failures, a node never sends different messages to different nodes.
Hence, a simple two phase blockchain protocol would be a round based design in which a deterministically selected leader disseminates its block proposal to all nodes using WRB.
In case all nodes deliver the proposed block, then this becomes the next block.

Given the continuous iterative nature of blockchains, we may improve the algorithm's round latency to an amortized single phase.
For this, we piggyback the ${r+1}^{th}$ block on top of the first message that \emph{WRB-deliver} sends when trying to deliver the $r^{th}$ block.
We support this by augmenting the \emph{WRB-deliver} method to receive two parameters, \emph{WRB-deliver(k,pgd)}, such that \emph{pgd} is the potential piggybacked data (can be $\mathit{nil}$).
The OBBC protocol is also augmented to receive \emph{pgd} and piggyback it on the first message it broadcasts.
Recall that we assume an OBBC protocol that always starts by having each node broadcast its vote.
Hence in the augmented protocol, each node that starts the OBBC protocol broadcasts its vote alongside the piggybacked \emph{pgd} message, which is made available to the calling code together with the decision value.

%\subsubsection{The CFT \PROTNAME{} Algorithm}
%\begin{figure}[t]

%\end{figure}

%The regular numbered lines in algorithm~\ref{cft-bft} presents a leader based, single-phased CFT ordering algorithm for blockchains.

The details appear in Algorithm~\ref{cft-bft}. Specifically, on each round $r$ the algorithm performs the following:
\begin{itemize}
	\item If $p_i$ is $({r+1})$'s proposer, it prepares a new block (lines~\ref{cft:l3}-\ref{cft:l4}).
	To ensure liveness, if in the previous iteration WRB has failed to deliver a message, then $r$'s proposer also prepares a block and WRB-broadcast it (lines~\ref{cft:l1}-\ref{cft:l2}).
	\item Meanwhile, all nodes are trying to WRB-deliver $r$'s block (line~\ref{cft:l5}).
	Note that the ${r+1}^{th}$ proposer piggybacks the next block on top of this message.
	\item If a block $b^r\neq nil$ has been delivered, then $b^r$ is appended to the chain (line~\ref{cft:l6}) and the protocol continues to round $r+1$ (line~\ref{cft:l8}-~\ref{cft:l9}).
	\item Else, all correct nodes switch proposer and continue to the next try (lines~\ref{cft:20}-~\ref{cft:l7}).
\end{itemize}
We prove the correctness of the full BFT protocol (Section~\ref{bft:proof}).
\begin{figure}[t]
	\centering
	\iftoggle{SINGLE}
	{\includegraphics[width=0.7\linewidth]{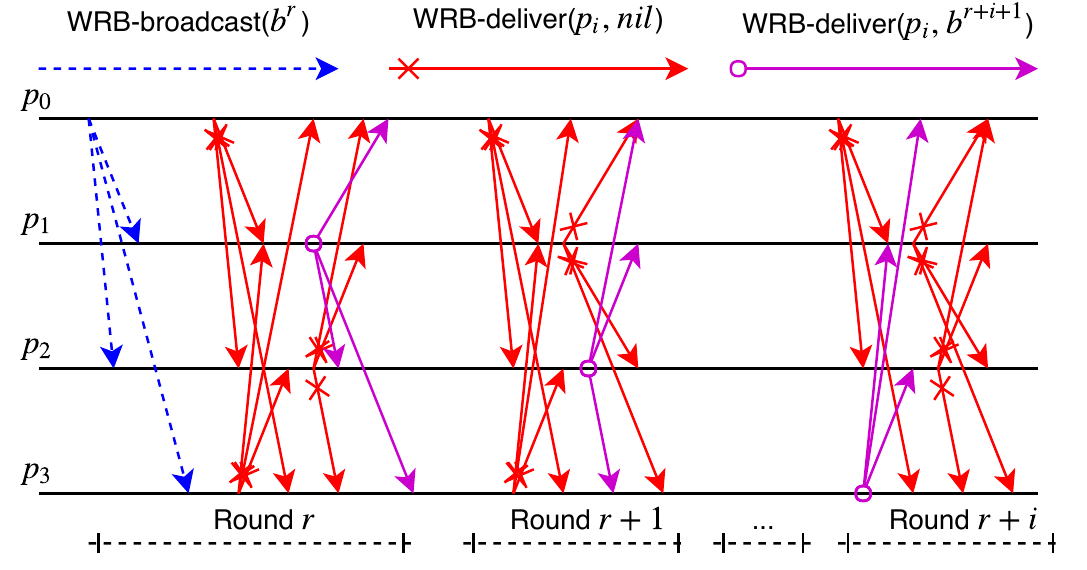}}
	{\includegraphics[width=1\linewidth]{figures/normal.pdf}}
	\caption{\PROTNAME{} in the normal case operation.}
	\label{fig:normal_case}
\end{figure}
Note that due to WRB and the piggybacking method, algorithm~\ref{cft-bft} establishes a single-phased protocol as long there are no Byzantine failures, and the failure pattern matches the specific OBBC optimistic pattern.

Figure~\ref{fig:normal_case} presents normal case operation. 
Each optimistic period starts with the current proposer broadcasting its block.
Then, on every round, each node broadcasts a single message (as the first OBBC message), except the next proposer that piggybacks the next block on top of that message.

\iftoggle{VLDB}{
	\begin{algorithm}[t]
	%	\caption{\PROTNAME{} -- code for $p_i$\\ \footnotesize The black and red lines depict the CFT and BFT mechanisms respectively}
	%	\label{cft-bft}
	\footnotesize
	\SetAlgoLined
	\SetKwFunction{ABB}{Atomic-broadcast}
	\SetKwFunction{ABD}{Atomic-deliver}
	\SetKwFunction{REC}{RECOVERY}
	
	\setcounter{AlgoLine}{0}
	\SetNlSty{texttt}{[}{]}
	\SetKwProg{myproc}{Procedure}{}{}
	\myproc{\REC{$r$, $proof$}}{
		$\mathit{versions_r}\leftarrow \{\}$\label{bft:l8}\;
		\eIf{$r_i < r - 1$}
		{
			$v \leftarrow \mathit{empty\_version}$\;
		}{
			$v \leftarrow [(b^{r - (f + 1)}, sig_{\mathit{proposer_{r - (f + 1)}}}(b^{r - (f + 1)})), ...,
			\newline (b^{r - 1}, sig_{\mathit{proposer_{r - 1}}}(b^{r - 1})), ..., \newline (b^{r_i}, sig_{\mathit{proposer_{r_i}}}(b^{r_i}))]$\;
		}
		\ABB{$v$}\;  \label{bft:l9}
		\Repeat{$|\mathit{versions_r}| = n - f$}
		{\label{bft:l10}
			\ABD $v_j$ from $p_j$\; 
			\If{$v_j$ is $valid$} {
				\tcc{validating the received version} 
				$\mathit{versions_r} \leftarrow \mathit{versions_r} \cup \{v_j\}$\;
			} 
			
		}
		$v' \leftarrow $ the first received among $\{v_j \in \mathit{versions_r} |r_j = max\{r_k | v_k \in \mathit{versions_r} \land (b^{r_k}, sig_{\mathit{proposer_{r_k}}}(b^{r_k}) \in v_k\} \}$\;
		adopt $v'$ and update $r_i$, and $\mathit{proposer_{r_i}}$\;
		$\mathit{full\_mode \leftarrow true}$;
	} \label{bft:l11}
	\caption{Recovery Procedure -- code for $p_i$}
	\label{recovery}
\end{algorithm}
}{}

\subsection{Full BFT \PROTNAME{}}
\nottoggle{VLDB}{In order to}{To} extend the basic \PROTNAME{} to handle Byzantine failures, \PROTNAME{} utilizes the fact that there is at least one correct node $p_c$ in every $f+1$ different proposers.
Since $p_c$ is correct, when $p_c$'s block is WRB-delivered, all correct nodes receive the very same block, including the hash to its predecessor block.
When a correct node detects a chain inconsistency (due to the hash that each block carries), it initiates a traditional BFT based recovery procedure.
At the end of the recovery phase, all correct nodes synchronize their chain to the same single valid version.
\nottoggle{VLDB}{To enable}{For} the recovery procedure, \PROTNAME{} maintains the following invariant:
\begin{invariant}
	\label{inv:continue}
	A node $p$ proceeds to the next round of the algorithm only if it knows that at least $f+1$ correct nodes will eventually proceed as well.
\end{invariant}
An immediate consequence of Invariant~\ref{inv:continue} is that if a block $b^r$ is at depth $f+1$ in $p$'s local chain, then there are at least $f+1$ correct nodes for which $b^r$ is at depth of at least $f+1$ in their local chains.
In principle, preserving Invariant~\ref{inv:continue} requires waiting to verify that at least $f$ other correct nodes are moving to the next round.
Yet, \nottoggle{VLDB}{since}{as} the communication pattern we described above already includes an all-to-all message exchange in each round (while executing OBBC), it serves as an implicit acknowledgement, so we \nottoggle{VLDB}{still}{} have a single-phased algorithm when the optimistic assumptions hold.

Recall that line numbers prefixed by `b' Algorithm~\ref{cft-bft} describe the additional actions to accommodate Byzantine failures.
Specifically:
\begin{itemize}
	\item First $p_i$ finds, by a pre-defined order, a prosper that has not successfully proposed a block in the last $f+1$ rounds (lines b\ref{bft:l1}- b\ref{bft:l2}).
	\item If $p_i$ detects an inconsistency in the chain, it \nottoggle{VLDB}{announces a recovery procedure}{starts the recovery} using reliable broadcast (lines b\ref{bft:l3}- b\ref{bft:l4}).
	\item Upon receiving a valid announcement of inconsistency, $p_i$ initiates the recovery procedure (lines b\ref{bft:l5}- b\ref{bft:l6}).
	Note that the announcement validation is done against digital signatures and the blocks' hashes.
\end{itemize}
The recovery \nottoggle{VLDB}{procedure}{} installs agreement among all correct nodes regarding the longest possible blockchain prefix as detailed in Algorithm~\ref{recovery}.
Executing \textsc{RECOVERY} by $p_i$ involves:
\begin{itemize}
	\item $p_i$ proposes, using \emph{Atomic-broadcast}, a valid version of the $f$ blocks that are in disagreement (excluding $b^r$ itself) followed by all the newer blocks it knows about (lines~\ref{bft:l8}--\ref{bft:l9}).
	\item Then $p_i$ collects $n-f$ valid versions (including empty ones) and adopts the first longest agreed prefix of the blockchain (lines~\ref{bft:l10}--\ref{bft:l11}).
\end{itemize}
Finally, as the recovery procedure may alter only the last $f+1$ blocks, the node decides on the block which is in depth of $f+2$ (line~b\ref{bft:l12})
\iftoggle{FV}{
	\subsection{Correctness Proof}
	\label{bft:proof}
	\begin{lemma}[\ACRDEF{$f+1$}-Validty]
		\label{proof:lemma1}
		A decided (possible tentative) value $v$ satisfies the \textsc{valid} method.
	\end{lemma}
	\begin{proof}
		From Algorithm~\ref{cft-bft}'s code, a value is appended to the blockchain only if it satisfies the \textsc{valid} method.
	\end{proof}
	\begin{lemma}
		\label{proof:lemma2}
		Every $f+1$ consecutive decided blocks were proposed by $f+1$ different nodes.
	\end{lemma}
	\begin{proof}
		By Algorithm~\ref{cft-bft}'s code (lines~b\ref{bft:l1}--b\ref{bft:l2}), if the current proposer has successfully proposed a block in the last $f+1$ rounds, then its role is switched to a new proposer.
	\end{proof}
	\begin{lemma}
		\label{proof:lemma3}
		At the recovery procedure's end, all correct nodes adopt the same version.
	\end{lemma}
	\begin{proof}
		By the \textbf{Atomic-Order} property of the atomic broadcast primitive all correct nodes receive the same versions in the same order. Hence, applying the deterministic rule described in Algorithm~\ref{cft-bft} results in an agreement among the correct nodes.
	\end{proof}
	\begin{lemma}
		\label{proof:lemma4}
		For any two correct nodes $p,q$, let $b_p^r, b_q^r$ be their decided value in round $r$.
		If $d(b_p^r) > f \wedge d(b_q^r) > f$ then $b_p^r=b_q^r$.
	\end{lemma}
	\begin{proof}
		Let $p,q$ be two correct nodes and let $b_p^r, b_q^r$ be their decided value in round $r$ such that $d(b_p^r) > f \wedge d(b_q^r) > f$ and assume b.w.o.c. that $b_p^r\neq b_q^r$.
		
		Let $p_r$ be $r$'s proposer.
		If $p_r$ is correct, then by \textbf{WRB-Agreement} and $p_r$'s correctness all correct nodes have received the very same message from $p_r$ (as in Reliable Broadcast).
		Hence, the case of $b_p^r\neq b_q^r$ may occur only as a result of the recovery procedure invocation in round  $r' \in \{r+1, ..., r + f\}$.
		By Lemma~\ref{proof:lemma3} at the end of the recovery procedure all correct nodes agree on the same version. Hence, $b_p^r = b_q^r$, contrary to the assumption.
		
		Suppose that $p_r$ is Byzantine and has disseminated different blocks to $p$ and $q$.
		Then, by Lemma~\ref{proof:lemma2} and the fact that the system model assumes at most $f$ faulty nodes there is at least one block $b^{r'} \in  \{b^{r+1}, ..., b^{r + f}\}$ that was proposed by correct proposer $p_{r'}$.
		Due to $p_{r'}$'s correctness it disseminates the same block to all, including the associated authentication data.
		Assume w.l.o.g. that $p_{r'}$ has received from $p_r$ the same version as $p$, then following Algorithm~\ref{cft-bft} (lines~b\ref{bft:l3}--b\ref{bft:l4}) $q$ receives an invalid block from $p_{r'}$ and invokes the recovery procedure.
		As $q$ is correct and by \textbf{RB-Termination} $p$ RB-delivers $q$'s proof of $b_{q}^{r'}$ invalidity and triggers the recovery procedure.
		Following Lemma~\ref{proof:lemma3} at the end of the recovery procedure all correct nodes adopt the very same version. Hence, $b_p^r = b_q^r$. A contradiction.
	\end{proof}
	
	\begin{lemma}[\ACRDEF{$f+1$}-Agreement]
		\label{proof:lemma5}
		For any two correct nodes $p,q$, let $b_p^r, b_q^r$ be their decided value in round $r$.
		If $d(b_p^r) > f+1 \wedge d(b_q^r) > f+1$ then $b_p^r=b_q^r$.
	\end{lemma}
	\begin{proof}
		Follows directly from Lemma \ref{proof:lemma4}.
	\end{proof}
	
	\begin{definition}
		\label{def:def1}
		Let $b_p^{r}, b_p^{r'}$ be two decided (possibly tentative) blocks of $p$ such that $r' > r$.
		$b_p^{r'}$ is \emph{valid with respect to} ${b_p^{r}}$ if the sub-chain $[b_p^r, b_p^{r+1}, ..., b_p^{r'}]$ satisfies the predefined \textsc{valid} method.
		Similarly, a sub-chain $[b_p^r, b_p^{r+1}, ..., b_p^{r'}]$ is \emph{valid with respect to} $b_p^{r - k}$, for some $ k \leq r$, if the sub-chain $[b_p^{r - k}, ..., b_p^{r}, ..., b_p^{r'}]$ satisfies the predefined \textsc{valid} method and each $f+1$ consecutive blocks are proposed by $f+1$ different proposers.
	\end{definition}
	
	\begin{lemma}
		\label{proof:lemma6}
		If during the recovery procedure for round $r$, a correct node $p$ receives a version $v$ from correct node $q$, then $v$ is valid with respect to $b_p^{r - (f + 2)}$.
	\end{lemma}
	\begin{proof}
		If the received version is an empty one, it is trivially valid with respect to $b_p^{r - (f + 2)}$.
		Else, by Lemma~\ref{proof:lemma4} all correct nodes agree on $b^{r -(f + 2)}$.
		As $q$ is correct, by Lemma \ref{proof:lemma1} $q$ appends only valid blocks to its blockchain.
		Hence, $q$'s version (that starts with $b_q^{r - (f + 1)}$) is valid with respect to $b_q^{r - (f + 2)}$ which is identical to $b_p^{r - (f + 2)}$.
	\end{proof}
	%\textbf{[[[Did we define what it means to be valid with respect to some block?]] {\color{red}better?}}
	
	\begin{definition}
		\label{def:def2}
		Let $r_g$ be the most advanced round of the algorithm that any correct node runs.
		We define the group of nodes whose current round is $\in\{r_g, r_g -1\}$ by $\mathit{front}=\{p|r_p$ $\in \{r_g-1, r_g\}\}$.
		When Invariant~\ref{inv:continue} is kept, there are at least $f+1$ correct nodes in \emph{front}.
	\end{definition}
	
	\begin{lemma}
		\label{proof:lemma7}
		While Invariant~\ref{inv:continue} is kept, a correct node executing the recovery procedure receives $n-f$ valid versions and at least one of them was received from a node in \emph{front}.
	\end{lemma}
	\begin{proof}
		By \textbf{RB-Termination} if a correct node $p$ detects an invalid block and invokes the recovery procedure, eventually every correct node will receive $p$'s proof and will invoke the recovery procedure.
		Following the system model, the ratio between $n$ and $f$ and Lemma~\ref{proof:lemma6}, a correct node does not get blocked while waiting for $n-f$ valid versions (part of whom may be empty). By Invariant~\ref{inv:continue}, $p$ receives at least one version from a correct node in $\mathit{front}$.
	\end{proof}
	
	\begin{lemma}[\ACRDEF{$f+1$}-Termination]
		\label{proof:lemma8}
		Every round eventually terminates.
	\end{lemma}
	\begin{proof}
		Let $r$ be a round of the algorithm.
		As long WRB-deliver has returned $nil$, a correct proposer will propose its next block.
		Thus, by \emph{WRB-Termination(1)}, \emph{WRB-Termination(2)} and the algorithm's iterative nature every correct node eventually succeeds to \emph{WRB-deliver} $b^r \neq nil$.
		
		If $b^r$ is valid and the recovery procedure has not been invoked, then naturally nothing prevents the round from terminating.
		If the recovery procedure has been invoked, then following Lemma~\ref{proof:lemma7} every correct node eventually receives $n-f$ valid versions and finishes the recovery procedure; in the worst case, such a node stays in round $r-1$.
		The system model assumes at most $f$ faulty nodes, which means that the above can repeat itself at most $f$ times in a row.
		Hence, after at most $f$ consecutive invocations of the recovery procedure, a correct proposer proposes $b^r$, ensuring that $r$ terminates in all correct nodes.
	\end{proof}
	\begin{definition}
		\label{proof:def3}
		Let $r$ be a round of the algorithm, we define by $\textit{front}_r=\{p|r_p > r + f + 1\}$ the group of nodes whose current round is greater by at least $f+1$ \nottoggle{VLDB}{rounds}{} than $r$.
	\end{definition}
	\begin{lemma}[\ACRDEF{$f+1$}-Finality]
		\label{proof:lemma9}
		In every round $r' > r + f + 1$, $v_p^r$ is definite.
	\end{lemma}
	\begin{proof}
		Following Algorithm~\ref{cft-bft}'s code, after a valid block is WRB-delivered, the only way it can be replaced is by an invocation of the recovery procedure.
		By Algorithm~\ref{cft-bft}'s code and Lemma~\ref{proof:lemma7}, if the recovery procedure has been invoked regarding round $r$, then $p$ adopts a prefix version that was suggested by $q \in \textit{front}$.
		Obviously, $\textit{front} \subseteq \textit{front}_r$. 
		By Lemma~\ref{proof:lemma5}, $\forall q,p\in \textit{front}_r, b^r_q = b^r_p$.
		In other words, $b^r_p$ is definite.
		
		Note that the ratio between $n$ and $f$ as well as the version's validation test ensure that no Byzantine node is able to propose a version that does not include $b^r_q$. 
	\end{proof}
\nottoggle{VLDB}
{
	Notice that due to asynchrony, it is possible that some node $q \in \mathit{front}$ advances to round $r+1$ while another node $p \in \mathit{front}$ invokes the recovery procedure for round $r+1$.
	In such a case, $q$ may adopt a new blockchain prefix which was proposed by $p$ and thus replace $b_q^{r - (f+1)}$.
	This edge case, illustrated in Figure~\ref{fig:edgecase}, is the reason why \PROTNAME{} implements \ACRDEF{$f+1$} rather than \ACRDEF{$f$}.
	\begin{figure}[t]
		\centering
		\iftoggle{SINGLE}
		{\includegraphics[width=0.7\linewidth]{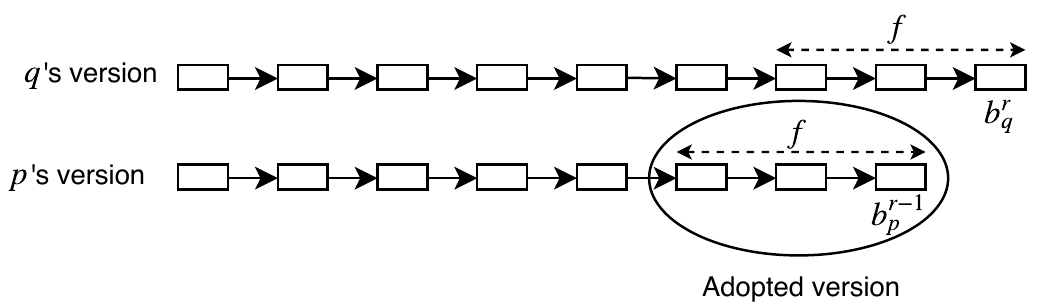}}
		{\includegraphics[width=1\linewidth]{figures/edge.pdf}}
		\caption{If \PROTNAME{} would implement \ACRDEF{$f$} instead of \ACRDEF{$f+1$}, a node $q$ may replace a definite block.
			While recovering, due to contention on $b_q^r$, the nodes may adopt a prefix that was suggested by $p \in \mathit{front}$ while $p$ is still in round $r$.}
		\label{fig:edgecase}
	\end{figure}
}{}
	
\nottoggle{VLDB}
{
	The next theorem follows from Lemmas~\ref{proof:lemma1}, \ref{proof:lemma5},~\ref{proof:lemma8} and~\ref{proof:lemma9}:
}
{
		From Lemmas~\ref{proof:lemma1}, \ref{proof:lemma5},~\ref{proof:lemma8} and~\ref{proof:lemma9} we have:
}
	\begin{theorem}
		\label{proof:theorem1}
		The protocol listed in Algorithm~\ref{cft-bft} solves \ACRDEF{$f+1$}.
	\end{theorem}
	
	\begin{lemma}[Atomic Order]
		\label{proof:lemma10}
		All blocks delivered by correct nodes using Algorithm~\ref{cft-bft} are delivered in the same order.
	\end{lemma}
	\begin{proof}
		Follows directly from Lemma~\ref{proof:lemma5} and the iterative nature of the algorithm.
	\end{proof}
	
	By Lemma~\ref{proof:lemma10} and Theorem~\ref{proof:theorem1}, we have:
	
	\begin{theorem}
		The protocol listed in Algorithm~\ref{cft-bft} impose a total order of all blocks.
	\end{theorem}
	
	\begin{theorem}
		\label{thm:non-triv}
		The protocol listed in Algorithm~\ref{cft-bft} satisfies Non-Triviality.
	\end{theorem}
	
	The proof of Theorem~\ref{thm:non-triv} follows directly from Lemmas~\ref{proof:lemma8} and~\ref{proof:lemma9}, as well as the observation that correct nodes propose non-empty blocks whenever they hold clients' transactions that have not been included in any previously decided block.
	
}{
	The correctness proof for Algorithm~\ref{cft-bft} is fairly technical. 
	Due to space limitations, we defer it to the full version of this paper~\cite{BF2019full}.
}

\subsection{Theoretical Bounds and Performance}
\label{performance}
\begin{table}[t]
\iftoggle{VLDB}{\footnotesize}{}
	\centering
	\begin{tabu} to 1\linewidth {||c|c|c|c||}
		\hline
		& \makecell{fault-\\free} & \makecell{Timing,\\and omission\\failures} &
		\makecell{Byzantine\\failures}\\
		\hline
		\hline
		\makecell{Communi-\\cation\\steps} & $1$ & $2+OBBC$ & \makecell{RB\\$+n$ parallel\\AB}\\
		\hline
		\makecell{Digital\\signature\\operations} & $1$ & $OBBC$ & \makecell{RB$+$AB\\$+(n -f)$\\$\cdot\text{(chain size)}$}\\
		\hline
		\makecell{Latency\\(in rounds)} & $f+1$ & \makecell{no additional\\overhead} & \makecell{no additional\\overhead}\\
		\hline
	\end{tabu}
\iftoggle{VLDB}{\normalsize}{}
	\caption{Performances characteristics of \PROTNAME. 
		RB and AB stand for Reliable and Atomic Broadcast respectively. 
		The first column shows the performance in the fault-free synchronized case. 
		The second column depicts the additional costs in an unsynchronized period and omission failures. 
		The third column depicts the additional expenses in the presence of a non-benign fault.}
	\label{performance:table1}
\end{table} 
Table~\ref{performance:table1} summarizes the performance of \PROTNAME{} in each of its three modes.
In case of no failures and synchronized network, \PROTNAME{} performs in the OBBC of WRB-deliver a single all-to-all communication step as well as one digital signature operation.
Further, only one node broadcasts more than one bit.
This can be obtained since the latency for a definitive decision can be up to $f+1$ rounds.
%\iftoggle{FV}{We show in Section~\ref{sec:impl} that this tradeoff is quite reasonable.}{}
In case a message is not received by WRB-deliver due to timing, omission or benign failures, the algorithm runs the non optimistic phase of OBBC as well as two more communication steps (one-to-all and one-to-one) in which a node asks for the missed message from the nodes who did receive it.
Also, the amount of digital signature operations depends on the specific OBBC implementation that is used by WRB-deliver.
Finally, in case of Byzantine failures, \PROTNAME{} runs the recovery procedure which depends on the Reliable Broadcast and the Atomic Broadcast implementations.
\nottoggle{VLDB}{
	
	Note that one can trade additional communication steps for fewer digital signature operations by using signature-free implementation of \PROTNAME's base protocols, e.g,~\cite{Mostfaoui:2015, Friedman:2005, Correia:2006}.
	As \PROTNAME{} is inherently not signature-free, a node may authenticate up to $n-f$ received prefix versions.
}{}
Yet, when Byzantine failures manifest, a node does not lose blocks so the latency in terms of rounds remains the same.

	\section{Implementation}
\label{sec:impl}
\nottoggle{VLDB}{
\subsection{\PROTNAME's Implementation}
\label{sec:toyimpl}
We now describe \PROTNAME's implementation components as well as a few necessary performance optimizations.
}{}
\nottoggle{VLDB}{
\subsubsection{Optimizations}
\label{impl:opt}
}{
\subsection{Optimizations}
\label{impl:opt}
}
\iftoggle{VLDB}
{
\paragraph*{Separating Headers and Blocks}
}
{
\paragraph{Separating Headers and Blocks}
}
\PROTNAME's protocol enables to easily separate the data path from the consensus path, such that only block headers need to pass through the consensus layer while the block itself is being sent asynchronously in the background.

In practice, a node $p$ broadcasts a block as soon as the block is ready.
On $p$'s next proposing round, $p$ WRB-broadcasts a header of a previously sent block.
Respectively, upon WRB-delivering a header, if $p$ did not receive the block, it votes against delivering it (See Algorithm~\ref{alg:wrb}, lines~\ref{wrb:l4}--\ref{wrb:l5}).
In addition, if a decision was made to deliver a header, but the block itself has not been received by $p$, then $p$ has to retrieve the block from a correct node $q$ that has it.
Such a $q$ exists because the decision to deliver is done only if at least one correct node has voted in favor of delivering, which means that $q$ has the block.

%\iftoggle{FV}{
\iftoggle{VLDB}
{
\paragraph*{Dynamically Tuning the Timeout}
}
{
\paragraph{Dynamically Tuning the Timeout}
}
%WRB-deliver utilizes a \emph{timer} variable which implements the core mechanism of \PROTNAME's liveness.
%Recall that a timer expiration while node $p$ did not receive a message, may cause $p$ to run a recovery phase which is significantly longer than \PROTNAME's optimistic phase.

%Thus, tuning the timer introduces a simple tradeoff.
%Denote by $\tau$ the current timer value and by $R$ the amount of time that the recovery phase may last.
%Hence, in the worst case, WRB-deliver may last at least $d=\tau + R$ time.  
%In other words, tuning the timer to a strictly high value will decrease the probability of invoking the recovery phase, but also, when such an event will occur, the maximum delay $d$ may affect \PROTNAME's performance. 
%On the other hand, tuning the timer to a strictly small value will increase the probability to face the recovery phase, but respectively, on every occurrence of such event, \PROTNAME's will delay in about only $\tau + R \approx R$ time. 

To adjust the timer to the current network delays status, we dynamically adjust its value based on the exponential moving average (EMA) of the message delays over the last $N$ rounds.
Namely, denote by $d_k, timer_k$ the delay of a message and the timer of round $k$ respectively.
Then for every round $r$
\[timer_r= \frac{2}{N+1}\cdot \mathit{d_{r-1}} + timer_{r-2}\cdot(1-\frac{2}{N+1}).\]
A formal discussion of the above tuning model is out of the scope of this paper. 
%}

\iftoggle{VLDB}
{
\paragraph*{Benign FD}
}
{
\paragraph{Benign FD}  
}
\PROTNAME's algorithm enables implementing a simple benign failure detector (FD) such that a crashed node will not cause an unrestricted increase in the timer value.
Largely speaking, every node $p$ maintains a suspected list of the $f$ nodes to which $p$ has waited the most and above a predefined threshold. 
For every node $q$ in that list, on a WRB-deliver, $p$ does not wait for $q$'s message but rather immediately votes against delivering.
By the ratio between $n$ and $f$ there is at least one correct node $c$ that is not suspected by any correct node and thus the algorithm's liveness still holds.
Despite the above, if $c$ is one of the last $f$ proposers it would not be able to suggest a new block.
Hence, the suspected list is invalidate every time \PROTNAME{} is skipping a node that is in the last $f$ proposers (see Algorithm~\ref{cft-bft}, lines~\ref{bft:l1}--\ref{bft:l2}).
Also, if a Byzantine activity was detected, to avoid considering more than $f$ nodes as faulty, we invalidate the suspected list.

\nottoggle{VLDB}{
\paragraph{Consecutive Byzantine Proposers}
Recall that nodes are chosen by default to serve as the initial proposers of each block in a round-robin manner.
Hence, as presented so far, \PROTNAME's performance might suffer if multiple Byzantine nodes are placed in consecutive places in this round-robin order.
This is easily solvable by periodically changing the round-robin order to be a pseudo-randomly selected permutation of the set of nodes where the seed is unknon a-priory to the (Byzantine) nodes.
For example, it can be the result of a \emph{verifiable random function} (VRF)~\cite{Gilad:2017} whose seed is a given block's hash value.
In case the adversary is static, this can be done only once, or if the performance suddenly drops due to the activity of consecutive Byzantine nodes.
To overcome a dynamic adversary whose minimal transfer time is equal to the time required to agree on $k$ blocks, we can simply invoke the above every $k$ blocks.
}

\nottoggle{VLDB}{
\subsubsection{\PROTNAME's Instance Implementation}
}{
\subsection{\PROTNAME's Instance Implementation}
}
Figure~\ref{fig:toyarc} depicts the main components of a \PROTNAME's instance.
Using \PROTNAME's API one can feed the \emph{TX pool} with a new write request.
The \emph{main thread} creates a new block and \emph{WRB-broadcasts} it in its turn.
In addition, the \emph{main thread} tries to \emph{WRB-deliver} blocks relying on \emph{OBBC}.
If it succeeds, the block is added to the \emph{Blockchain}.
Meanwhile, the \emph{panic thread} waits for a panic message.
When such a message is \emph{Atomic-delivered}, the \emph{panic thread} interrupts the \emph{main thread} which as a result invokes the \emph{recovery} procedure.
\PROTNAME{} is implemented in Java and the communication infrastructure uses gRPC, excluding BFT-SMaRt which has its own communication infrastructure.
\emph{Atomic Broadcast} is natively implemented on top of \emph{BFT-SMaRt} whereas \emph{OBBC} uses \emph{BFT-SMaRt} only as the fallback mechanism when agreement cannot be reached through the optimistic fast path (which relies on gRPC).

\iftoggle{MYACM}{
	\begin{figure}[t]
		\centering
		\iftoggle{SINGLE}
		{\includegraphics[width=0.6\linewidth]{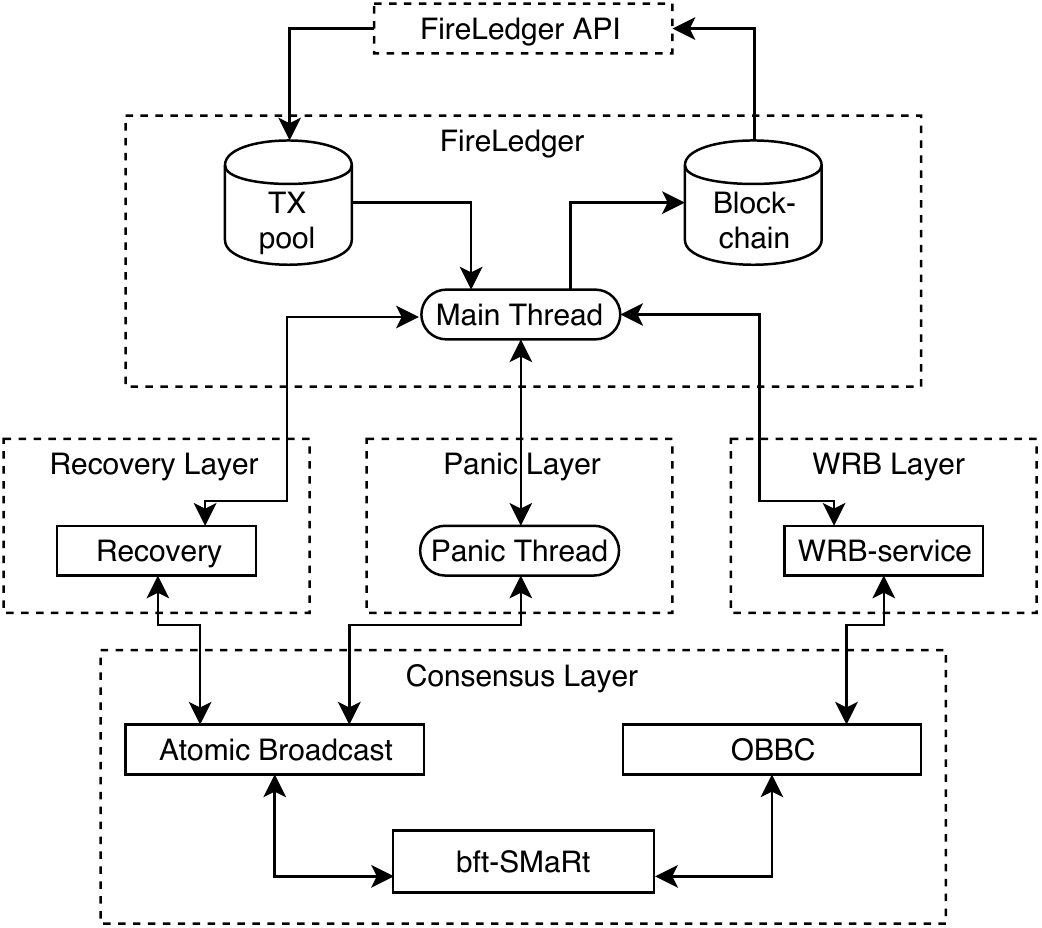}}
		{\includegraphics[width=0.9\linewidth]{figures/fl-arc-new.pdf}}
		%	\begin{fcap}
		\caption{An overview of \PROTNAME{} instance's main components}
		\label{fig:toyarc}
		%	\end{fcap}
	\end{figure}
}
{
		\begin{figure}[t]
		\centering
		\includegraphics[width=1\linewidth]{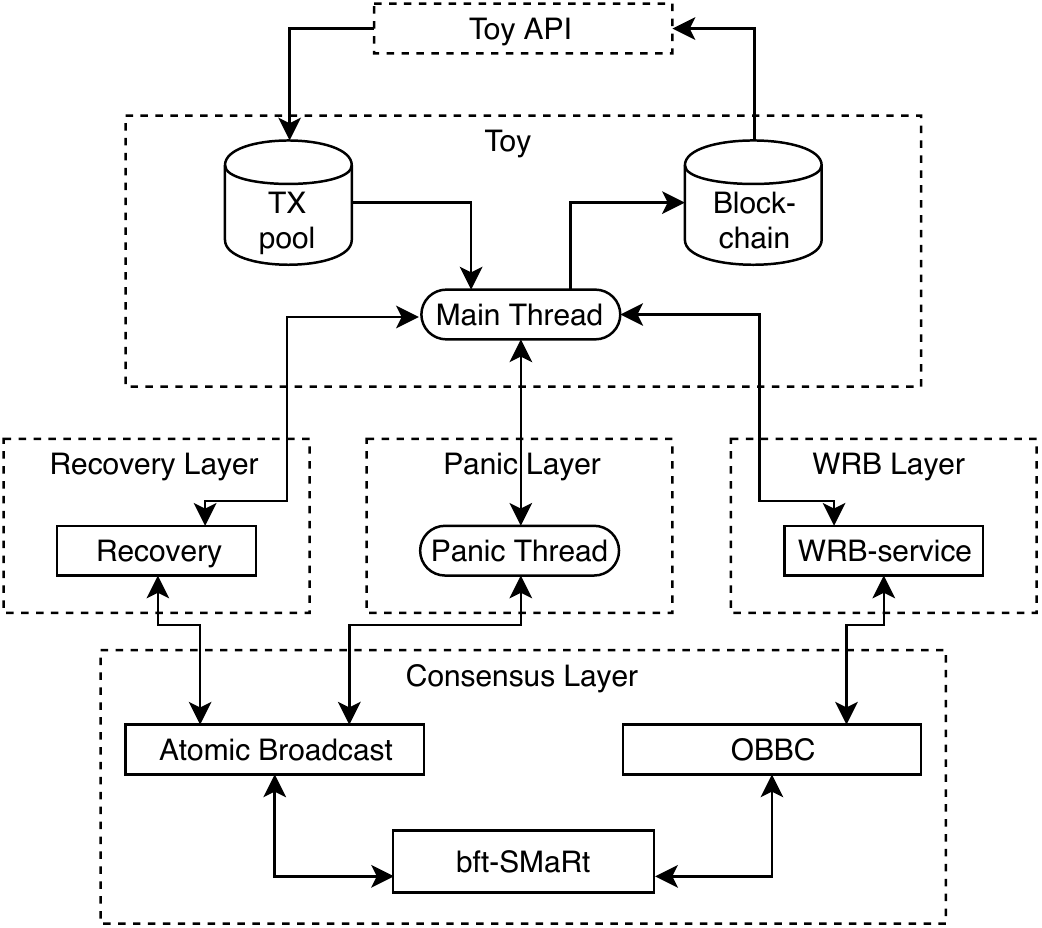}
		%	\begin{fcap}
		\caption{An overview of \PROTNAME{} instance's main components}
		\label{fig:toyarc}
		%	\end{fcap}
	\end{figure}
}

\iftoggle{MYACM}{
\subsection{\SYSNAME{} -- \PROTNAME{} Orchestrator}
\label{sec:topl}
}{
\subsection{\SYSNAME{} -- \PROTNAME{} Orchestrating Platform}
\label{sec:topl}
}

%We now present \SYSNAME, an orchestrator of \PROTNAME{s} instances.

%\subsection{Architecture}
%\label{sec:arch}
%A key issue when implementing a blockchain RSM system is how to maintain a high-throughput ordering service alongside with clients' low latency.
While \PROTNAME{} assumes partial synchrony, its rotating leader pattern imposes synchronization in the sense that a node may propose a value only on its turn, making \PROTNAME's throughput bounded by the actual network's latency.
To ameliorate this problem, we introduce another level of abstraction, named \emph{workers}, by which each node runs multiple instances of \PROTNAME{} and uses them as a blockchain based ordering service.
The use of workers brings two benefits:
($i$) workers behave asynchronously to each other which compensates for the above synchrony effect of \PROTNAME{} and ($ii$) while a worker waits for a message, other workers are able to run, resulting in better CPU utilization.
To preserve the overall total ordering property of \PROTNAME, a node must collect the results from its workers in a pre-defined order, e.g., round robin.
This requirement may imposes higher latency when the system is heavily loaded because even if a single worker faces the non-optimistic case, it delays all other workers from delivering their blocks to the node.

\iftoggle{MYACM}{
	\begin{figure}
		\centering
		\iftoggle{SINGLE}
		{\includegraphics[width=0.4\linewidth]{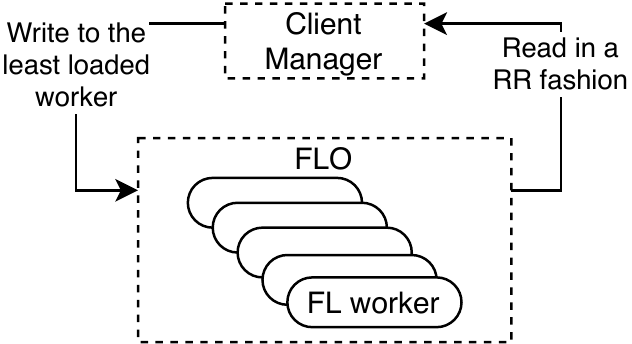}}
		{\includegraphics[width=0.7\linewidth]{figures/FLO-arc-new.pdf}}
		%	\begin{fcap}}
		\caption{An overview of a \SYSNAME{} node}
		\label{fig:arc}
		%	\end{fcap}
	\end{figure}
}
{
	\begin{figure}
		\centering
		\includegraphics[width=0.7\linewidth]{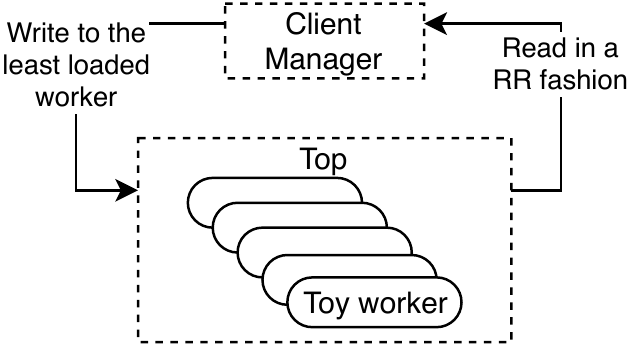}
		%	\begin{fcap}}
		\caption{An overview of a \SYSNAME{} node}
		\label{fig:arc}
		%	\end{fcap}
	\end{figure}
}

Figure~\ref{fig:arc} depicts an overview of a \SYSNAME{} node.
Upon receiving a write request, the \emph{client manager} directs the request to the least loaded worker.
When a read request is received, the \emph{client manager} tries to read the answer from the relevant worker. 
Only if the answer was already definitely decided and its block can be delivered in the pre-defined order, the node returns the answer\footnote{
This is similar to the \emph{deterministic merge} idea of~\cite{AS00}.}. 

\section{Performance Evaluation}
\label{sec:eval}
%We now describe the experiments conducted to evaluate \SYSNAME{} and discuss the results.
	\begin{table}[t]
		\iftoggle{VLDB}{\footnotesize}{}
		\centering
		\begin{tabu} to 1\linewidth {||c|c|c||}
			\hline
			\makecell{parameter} & \makecell{range} & \makecell{units} \\ 
			%& \makecell{motivation}
			\hline
			\hline
			cluster size & $n \in \{4, 7, 10\}$ & - \\
			%&\makecell{increasing the cluster size\\ affect the scalability}
			\hline
			\makecell{workers} & $1 \leq \omega \leq 10$ & - \\
			%&\makecell{multiple workers affect the\\throughput and latency} 
			\hline
			\makecell{transaction size} & \makecell{$\sigma \in \{512, 1K, 4K\}$} & Byte \\
			% & \makecell{affects the block size}
			\hline
			batch size & \makecell{$\beta \in \{10, 100, 1000\}$} & Transaction \\
			%& \makecell{affects the block size}
			\hline
		\end{tabu}
		\iftoggle{VLDB}{\normalsize}{}
		\caption{The default evaluation's parameters: The first row presents \SYSNAME's cluster size. The system model assumes $f < \frac{n}{3}$, hence, $n \in \{4, 7, 10\}$ imposes $f \in \{1, 2, 3\}$ respectively.}
		\label{impl:params}
	\end{table}

Our evaluation studied the following questions:
\iftoggle{VLDB}
{	($i$) Is \SYSNAME/\PROTNAME{} CPU bounded or network bounded?
	($ii$) How Table~\ref{impl:params}'s values affect \PROTNAME's performance?
	($iii$) How the node distribution method affects \PROTNAME's performance?
	($iv$) How does \PROTNAME{} handle failures?
}
{
\begin{itemize}
	\item Is \SYSNAME/\PROTNAME{} CPU bounded or network bounded?
	\item How Table~\ref{impl:params}'s values affect \SYSNAME/\PROTNAME{} performance?
	\item How the node distribution method affects \SYSNAME/\PROTNAME{} performance?
	\item How does \SYSNAME/\PROTNAME{} handle failures?
\end{itemize}
}
\iftoggle{VLDB}
{\paragraph*{Deployment Specification}}
{\paragraph{Deployment Specification}}
Our setup for most measurements includes $n$ nodes running on $n$ identical VMs with the following specification: 
m5.xlarge with $4$ vCPUs of Intel Xeon Platinum 8175 2.5 GHz processor, $16$ GiB memory and up to $10$ Gbps network links
(Section~\ref{sec:alternatives} uses a stronger configuration as detailed there).
In all measurements, transactions are randomly generated.

\subsection{Signature Generation}
In \SYSNAME/\PROTNAME{}, as long as the optimistic assumptions hold, a proposer signs its block only once and any other node is verifying the signature only once.
Hence, the maximal signature rate serves as an upper bound on the potential throughput of \PROTNAME.

To undrestand whether \SYSNAME/\PROTNAME{} is CPU bounded or I/O bounded we start by presenting an evaluation of the signatures generation rate (sps) which is typically the CPU most intensive task.
We use ECDSA signatures with the \emph{secp256k1} curve.
When signing a block, all the block's transactions are hashed and the result is signed alongside the block header.
We vary the $\omega, \beta$ and $\sigma$ values in the ranges described in Table~\ref{impl:params}.

\nottoggle{VLDB}{
Denote by $t_{hash}$ the time that takes to hash a single byte.
For a block consisting of $\beta$ transactions of $\sigma$ bytes each, the block's signing time, $t_{sign}$, is expected to be
\[
	t_{sign} = \beta \cdot \sigma \cdot t_{hash} + C
\]
where $C$ is a constant representing the time that takes to sign the fixed sized block header.
}{}

\begin{figure}[t]
	\centering
	\iftoggle{SINGLE}
	{\includegraphics[width=0.9\linewidth]{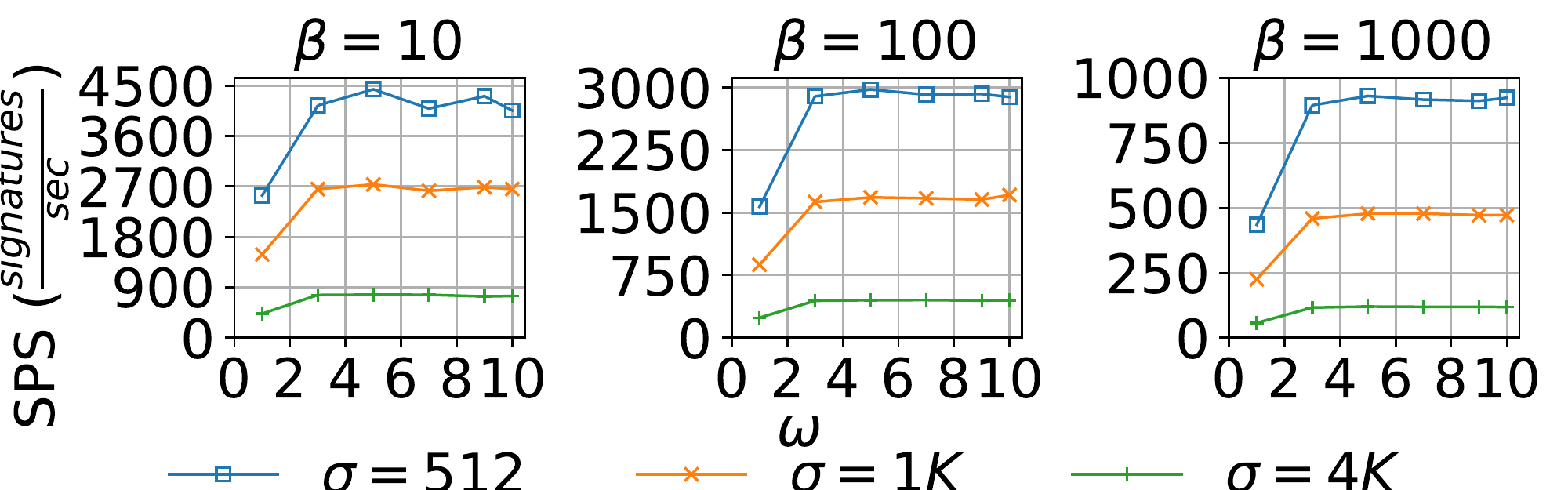}}
	{\includegraphics[width=1\linewidth]{figures/sigs.pdf}}
%	\begin{fcap}
	\caption{Signatures generation rate on a single VM}
	\label{fig:sigs}
%	\end{fcap}
\end{figure}
For each configuration we run the benchmark for 1 minute.
Figure~\ref{fig:sigs} depicts the benchmark results.
As expected, with small blocks the sps is higher than with larger ones.
Also, as our machines have 4 vCPUs, increasing $\omega$ beyond 4 has a minor effect if any.
\nottoggle{VLDB}{
	
Denote by $tps$ the transactions per second rate.
The following bound must hold:
\[ tps \leq sps \cdot \beta. \]}{}
As seen below, the performance of \SYSNAME{} is not limited by the $\mathit{sps}$ rate.

\subsection{\SYSNAME{} Cluster in a Single Data-Center}
We deployed a \SYSNAME{} cluster where all machines reside in the same data center.
To test the system in extreme conditions we simulate an intensive load by filling every block to its maximal size.
In practice, in every round, if a node does not have a full block to transmit, the node fills the block with random transactions, up to its maximal capacity, and then disseminates the full block.
 
In the current version, \PROTNAME{} uses a clique overlay for disseminating both blocks and headers.
To avoid clogging the network, \PROTNAME{} has a basic flow control mechanism that prevents nodes from sending new blocks if the network is overloaded or if they have already disseminated enough blocks that have not been decided yet.
Due to its modular design, a more sophisticated mechanism can be plugged into the system.
This is left for future work.

\subsubsection{\SYSNAME's Throughput}
\nottoggle{VLDB}
{To test \SYSNAME's throughput we measured the blocks per second (bps) and the transactions per second (tps) rates.}{}
For each configuration we run the experiment for 3 minutes. 
The results were collected from all nodes and we took the average among them.

\paragraph{BPS Rate} 
Due to the separation of blocks and headers, \SYSNAME's throughput is mostly bounded by its bps rate.
\begin{figure}[t]
	\centering
	\iftoggle{SINGLE}
	{\includegraphics[width=.5\linewidth]{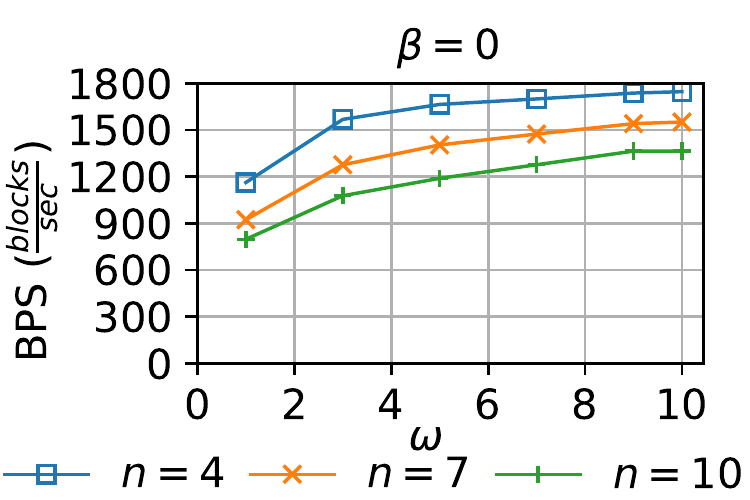}}
	{\includegraphics[width=.55\linewidth]{figures/bps2.pdf}}
	\caption{\SYSNAME's bps rate for $n \in \{ 4, 7, 10\}$ in a single data-center cluster}
	\label{fig:bps}
\end{figure}
Figure~\ref{fig:bps} presents \SYSNAME's bps rate for different $n$ and $\omega$ values.
As expected, increasing $\omega$ yields higher bps due to better CPU utilization. 
In contrast, increasing $n$ decreases the bps because each decision requires more communication.
Even so, \SYSNAME{} delivers thousands of bps under the majority of the tested configurations.
\iftoggle{VLDB}{
\SYSNAME's throughput is bounded by
$tps \leq \beta \cdot bps$.
}
{
Thus, for any configuration, \SYSNAME's throughput is bounded by
	 \[tps \leq \beta \cdot bps.\]
}
\paragraph{TPS Rate}
We tested \SYSNAME's throughput while varying $n, \omega, \sigma$ and $\beta$ values in the ranges described in Table~\ref{impl:params}.
\begin{figure}[t]
	\centering
	\iftoggle{SINGLE}
	{\includegraphics[width=0.9\linewidth]{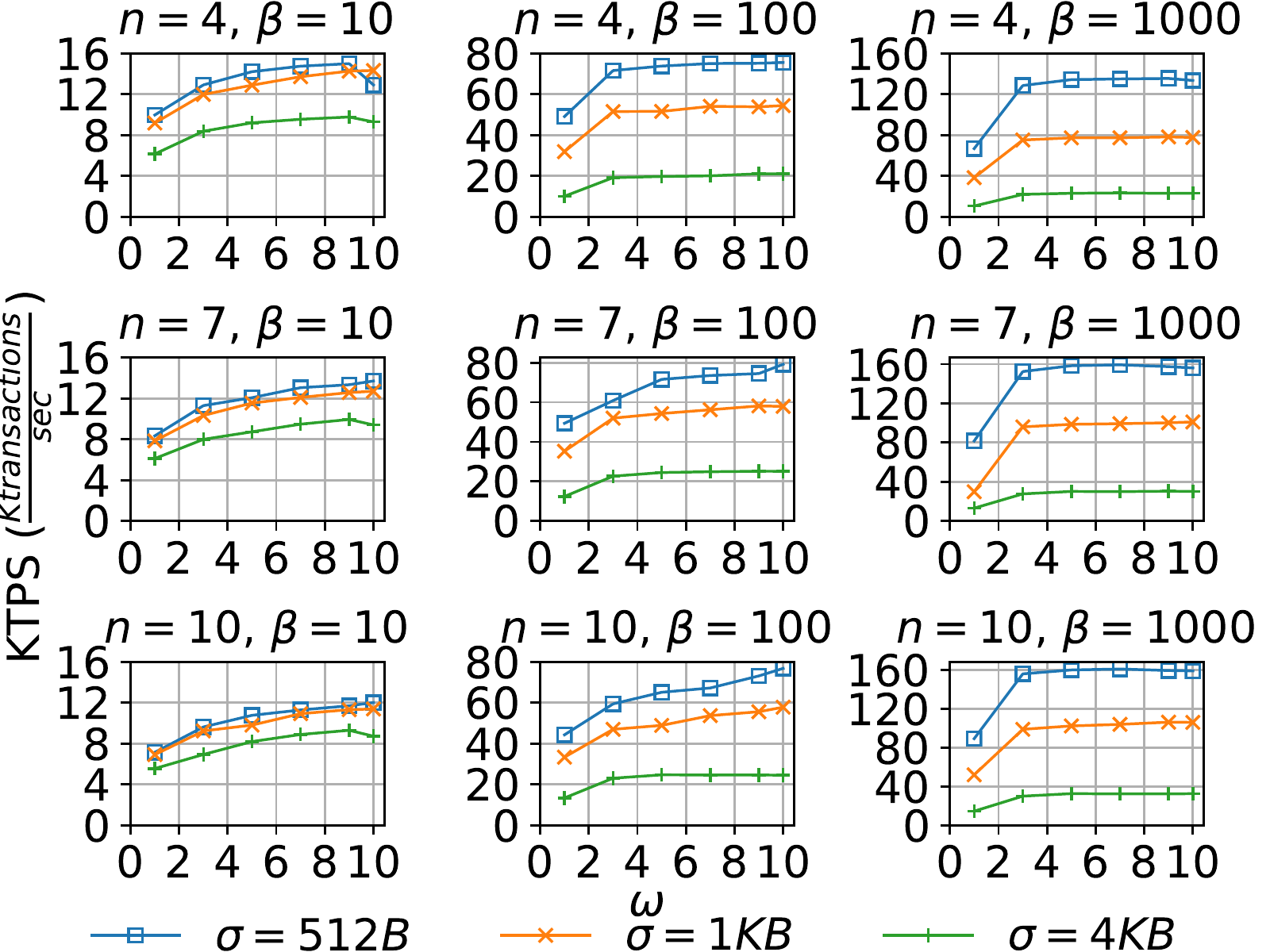}}
	{\includegraphics[width=1\linewidth]{figures/tps2.pdf}}
	\caption{\SYSNAME's transactions throughput for various configurations in a single data-center deployment}
	\label{fig:throughput}
\end{figure}
Figure~\ref{fig:throughput} shows \SYSNAME's throughput with the above configurations.
Except few configurations where $\beta=10$, \SYSNAME{} achieves between tens to hundreds of thousands of tps, depending on the specific configuration.
Especially, with $\sigma=512$, which according to~\cite{tradeblock} is the average size of a Bitcoin's transaction, \SYSNAME{} peaks at around 160K tps even with $n=10$.
As expected, for larger $\sigma$ the performance decreases because less of the network's bandwidth remains available for the headers, which limits the bps rate.
It can be observed that the performance for large blocks with $n \in \{7, 10\}$ is better than when $n=4$. 
This can be explained by the fact that the separation of blocks from headers allows for nodes in bigger clusters to collect more blocks that have not been decided yet.
Thus, as the bps grows w.r.t. the number of workers, it respectively increases the tps rate.
This does not manifest in small block sizes because with very small blocks there is little benefit from transmitting a block before its header. 
%\iftoggle{FV}{}{
%	
%In the full version of this paper~\cite{BF2019full} we measured the maximal signature generation rate in this setting and we are not bounded by it.
%}
%The maximal observed bandwidth is with $n=4, \omega=20, \sigma=1024$ and $\beta=1000$ and is about 90 MB/sec. 
%To evaluate \SYSNAME's performance w.r.t the hardware bounds, we tested the network's performance between two machines in the same datacenter. 
%We measured latency of 100 microseconds and bandwidth of 120 MB/sec.
%This means that \SYSNAME's throughput, in its peak, touches the network bounds.
\nottoggle{VLDB}{
The experiment results demonstrate that \SYSNAME's throughput is in line with the most demanding envisioned blockchain applications.
}{}

\subsubsection{\SYSNAME's Latency}
To evaluate \SYSNAME's latency, we measured the time it takes for a full block to be delivered by \SYSNAME.
This time includes disseminating the block, its header, and to wait until it can be delivered by the round robin between \SYSNAME's workers.
We focus on the configurations in which $\sigma = 512$ (same as Bitcoin transactions).
\begin{figure}[t]
	\centering
	\iftoggle{SINGLE}
	{\includegraphics[width=0.9\linewidth]{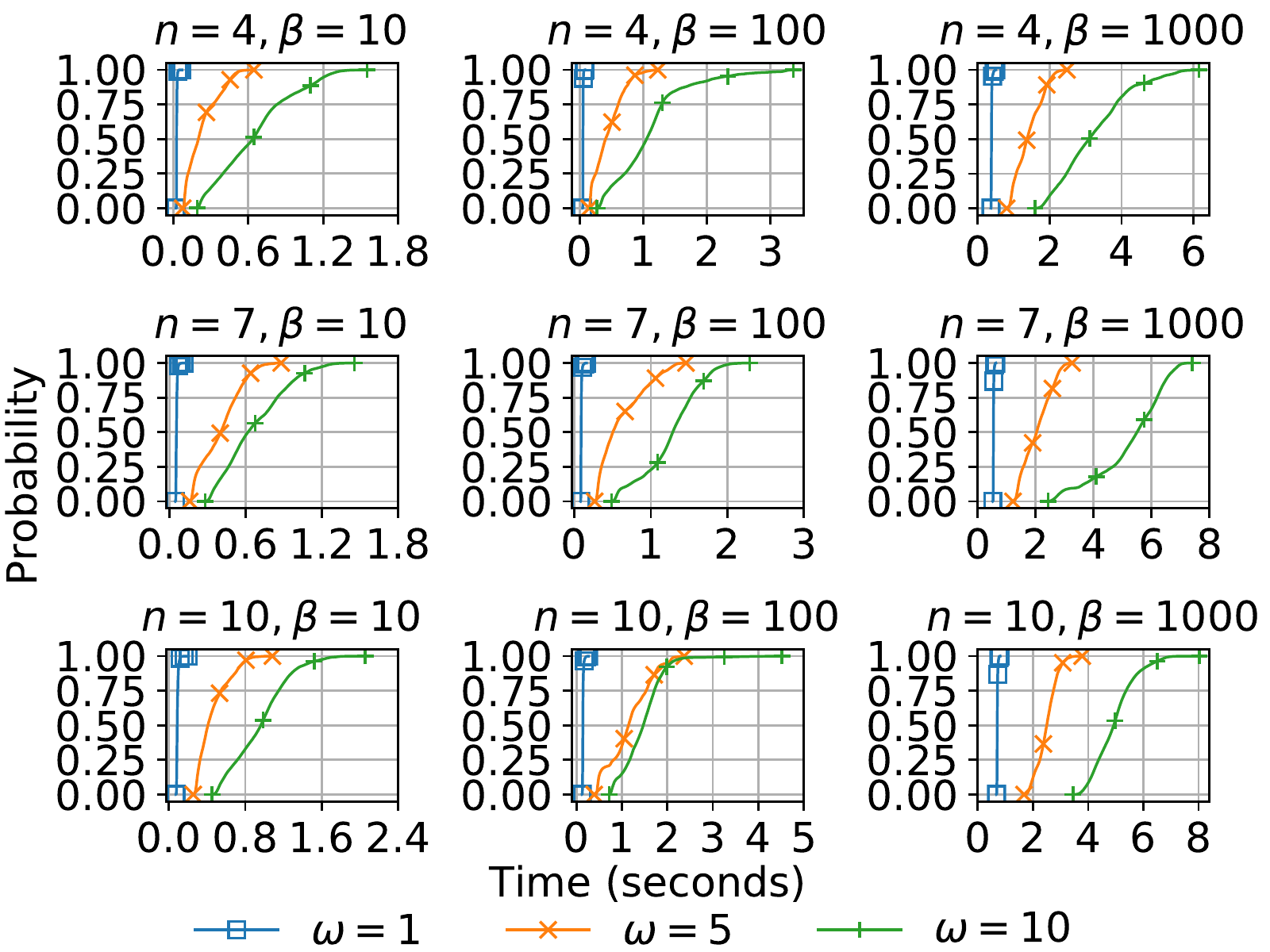}}
	{\includegraphics[width=1\linewidth]{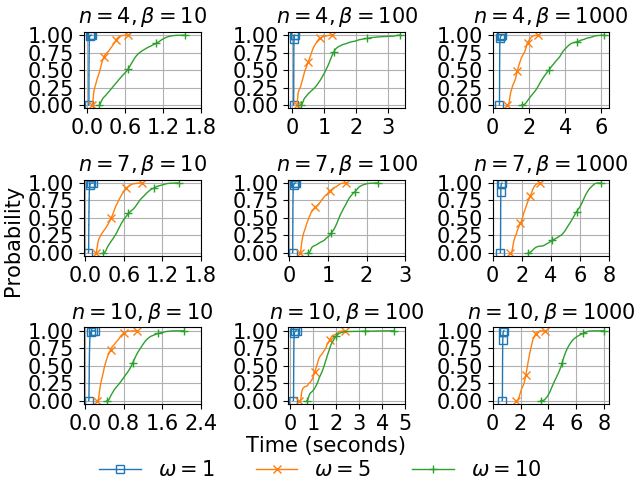}}
	\caption{CDF charts for $\sigma = 512, n \in \{4, 7, 10\}, \omega\in\{1, 5, 10\}$ and $\beta \in \{10, 100, 1000\}$ in a single data-center}
	\label{fig:cdfsinglecloud}
\end{figure}
Figure~\ref{fig:cdfsinglecloud} shows CDF charts for \iftoggle{VLDB}{}{the} various configurations. 
As expected, with $\omega=1$ \SYSNAME's latency is minimal and is less than 1 second even with $n=10$ and $\beta=1000$.
Increasing $\omega$ results in an increase in the latency respectively.
This is due to the fact that even a single worker's delay is reflected in all, due to \SYSNAME's round robin. 
Yet, even with $n=10, \omega=10$ and $\beta=1000$ the latency in a single data-center deployment is below 8 seconds.

\iftoggle{VLDB}{To}{In order to} understand the main bottlenecks of \SYSNAME, we divided every round of the algorithm into 5 different events:
(A) block proposal, (B) header proposal, (C) tentative decision, (D) definite decision, and (E) delivering by \SYSNAME.
Finally, we measured the time between each pair of consecutive events.
\begin{figure}[t]
	\centering
	\iftoggle{SINGLE}
	{\includegraphics[width=0.8\linewidth]{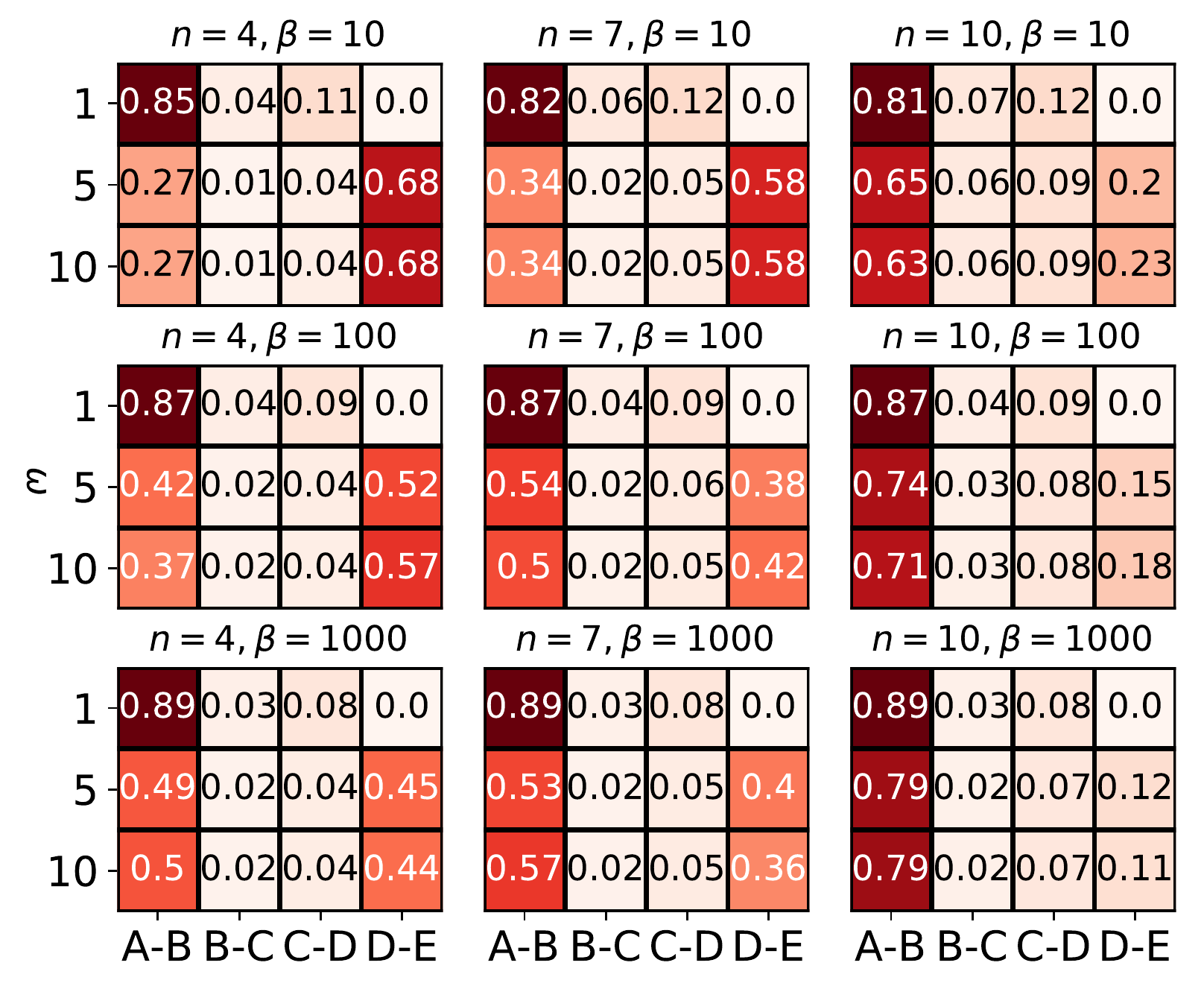}}
	{\includegraphics[width=1\linewidth]{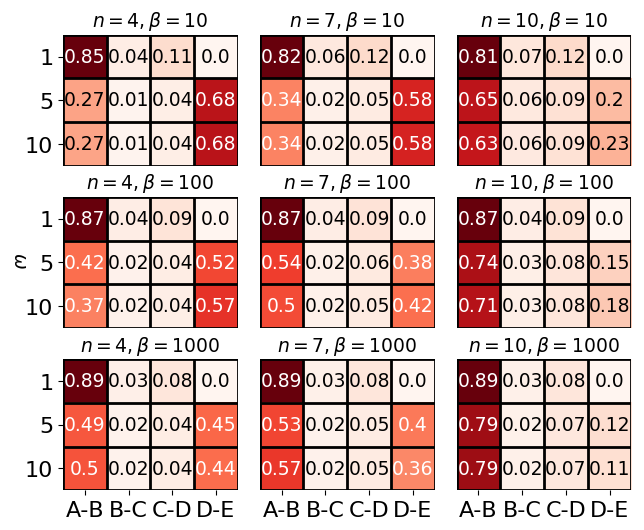}}
	\caption{Heatmaps that depicts the relative execution time of \SYSNAME{} between 5 different events for $\sigma=512$}
	\label{fig:heatmap2}
\end{figure}

Figure~\ref{fig:heatmap2} shows the relative execution time between each two consecutive events.
It is easy to see that due to the separation of blocks and headers, the majority of time is spent between receiving a block and receiving its header.
In addition, for $\omega > 1$, the workers cause an increase in the latency, as even a single worker's delay, delays the whole system.
Finally, increasing $n$ as well as $\beta$ causes the blocks dissemination event to take longer.
This is despite using a clique layout. %, which is the shortest time to disseminate data. {\color{red} indeed??}
Other methods (e.g., gossip) may improve the throughput but not the latency. 
 
\subsection{\SYSNAME's Scalability}
To test \SYSNAME's scalability we deployed a single data-center cluster of $n=100$ machines and tested \SYSNAME's tps rate with $\sigma=512, \beta \in \{10, 100, 1000\}$ and $1 \leq\omega \leq 5$.
We ran each configuration for 3 minutes.
\begin{figure}[t]
	\centering
	\iftoggle{SINGLE}
	{\includegraphics[width=0.5\linewidth]{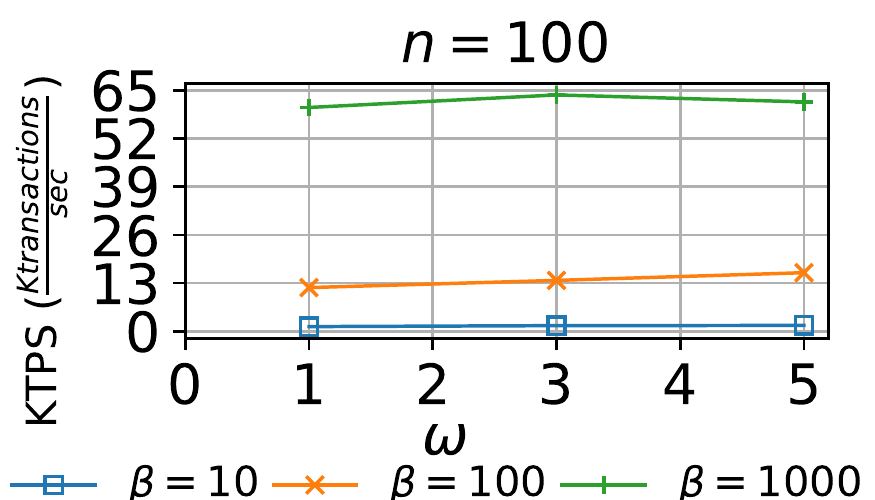}}
	{\includegraphics[width=0.55\linewidth]{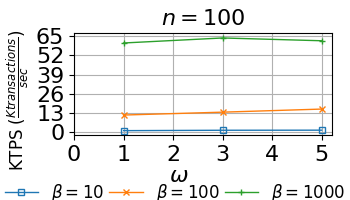}}
	\caption{\SYSNAME's tps rate with $n=100, \sigma=512, \beta\in\{10, 100, 1000\}$ and $1 \leq \omega \leq 5$}
	\label{fig:tpsbs}
\end{figure}
Figure~\ref{fig:tpsbs} depicts the benchmark's result.
Thanks to \PROTNAME's frugal communication pattern, as long as the fault free execution path take place, \SYSNAME{} can achieve around $60K$ tps (in a single data-center deployment).
\iftoggle{VLDB}{As shown, in this}{It can be seen that due to the} cluster size, the number of workers has no effect because of the relatively large amount of communication that even a single worker consumes. 

\subsection{\SYSNAME{} Under Failures} 
\subsubsection{Benign Failures}
We tested \SYSNAME{} while suffering from crash failures of $f$ nodes (yet, we maintain $n=3f+1$ even in this benign case).
Here, all faulty nodes crash in the middle of a run (such a node crashes with all of its workers), but the measurements are taken after the faulty nodes crash.
Every benchmark was ran for 3 minutes and we calculated the average tps among the correct nodes.
\begin{figure}[t]
	\centering
	\iftoggle{SINGLE}
	{\includegraphics[width=0.9\linewidth]{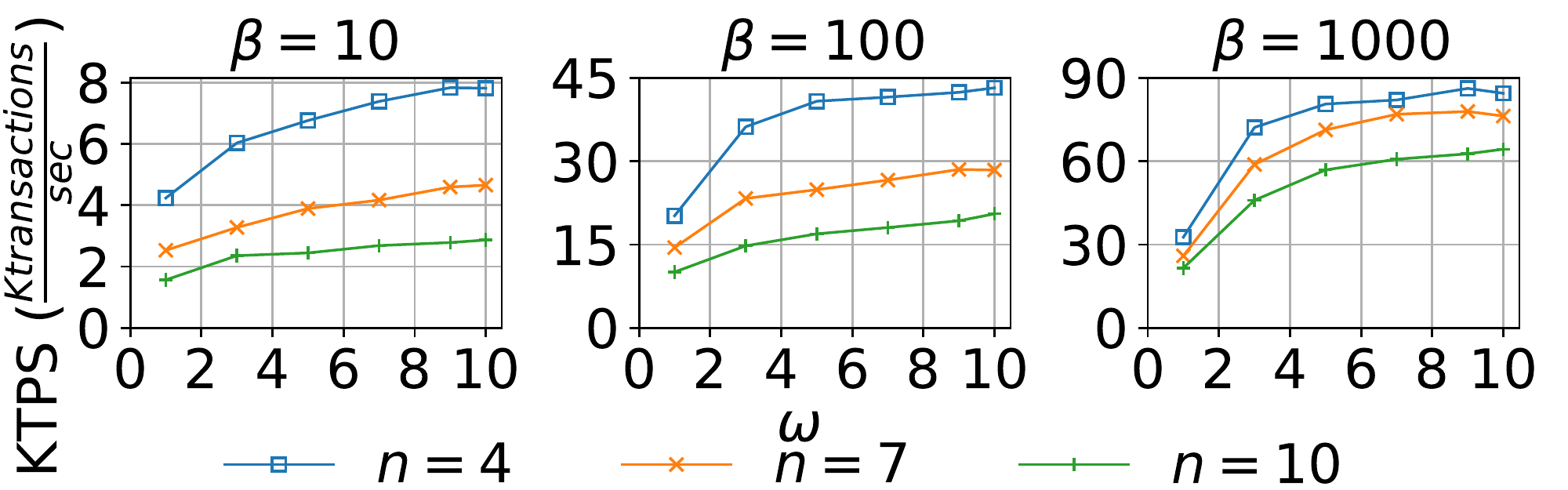}}
	{\includegraphics[width=1\linewidth]{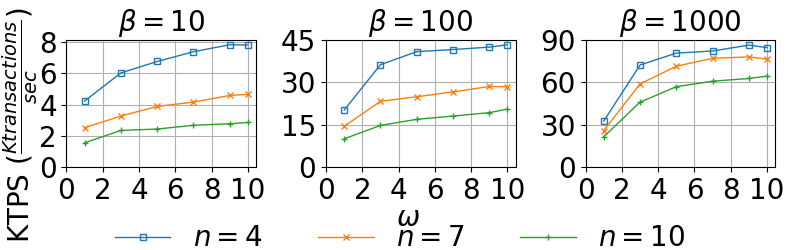}}
	\caption{\SYSNAME's tps rate while suffering from crash failures of $f$ nodes for $\sigma=512, \beta\in \{10, 100, 1000\}, n\in \{4, 7, 10\}$ and $f\in \{1, 2, 3\}$.}
	\label{fig:tpsbenign1}
\end{figure}
Figure~\ref{fig:tpsbenign1} depicts \SYSNAME's tps rate under various configurations while facing crash failures.

It can be seen that due to the full BBC phase that is now needed during faulty nodes rounds, for larger $n$ the tps is decreasing.
Yet, \SYSNAME{} still reaches tens of thousands of tps despite these benign failures. 
This is due to the OBBC protocol and the basic failure detector described in Section~\ref{impl:opt}.

\subsubsection{Byzantine Failures}
To test \SYSNAME's performance when facing Byzantine failures, we deployed a Byzantine \SYSNAME{} node that operates as following:
When started, every worker divides the cluster into two random parts, and for every given round it distributes different versions of the block to each part.
Notice that in practice, invoking the recovery procedure may cause the nodes to become un-synchronized with each other, a fact that affects the performance as well.
This increases the variance between measurements of the same settings.
Hence, to be able to perform more measurements of each data point, for each configuration we run a series of short benchmarks (between 1 - 2 minutes each).

\begin{figure}
	\centering
	\iftoggle{SINGLE}
	{\includegraphics[width=0.94\linewidth]{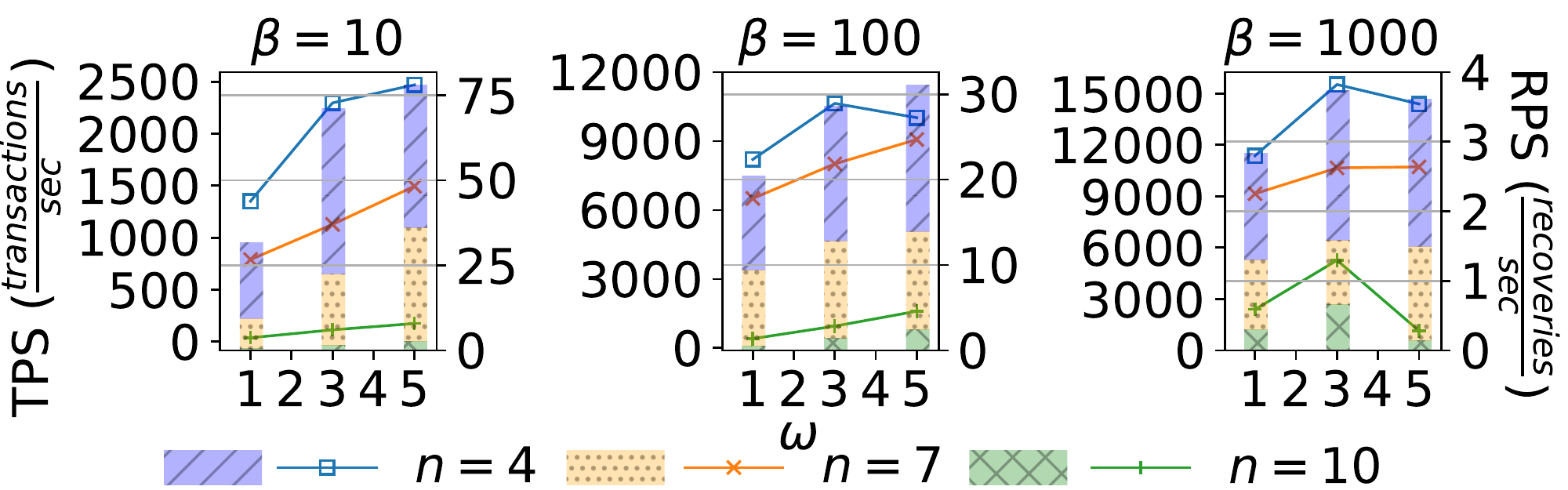}}
	{\includegraphics[width=1\linewidth]{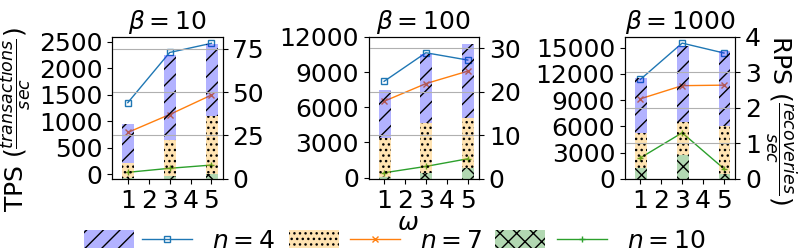}}
	\caption{\SYSNAME's tps rate under Byzantine failures for $\sigma=512, 1 \leq \omega \leq 5, \beta \in \{10, 100, 1000\}, n\in \{4, 7, 10\}$ and $f\in\{1, 2, 3\}$.
	The lines show the tps rate and the bars shows the recovery per second (rps) rate}
	\label{fig:tpsbyz1}
\end{figure}

Figure~\ref{fig:tpsbyz1} presents the tps rate for \SYSNAME{} when facing Byzantine failures w.r.t. the number of workers and the number of recoveries per second (rps).
Smaller values of $\beta$ and $n$ imply more recovery events. 
Recall that during the recovery nodes halt.
Thus, the above is expected due to the fact that each recovery ends faster when $\beta$ and $n$ are small.
Yet, for bigger $\beta$, the batching effect compensate for the small amount of recoveries and the long halts.
The reason why for $n=10, \beta=1000$ and $\omega=5$ the performance decreased so much is the underlying Byzantine consensus layer (BFT-SMaRt), which has to handle a large amount of data.
To conclude, \SYSNAME{} delivers more than 10K tps in some scenarios even when facing Byzantine failures.
Although the performance is lower than in optimistic executions, these type of failures are expected to be rare in permissioned blockchain clusters.
And even with $n=10$, if we set $\beta=1000$ and $\omega=3$ we obtain about $6$K tps, roughly twice the average tps of VISA.
Hence, \SYSNAME{} presents an attractive trade-off between scalability, performance and security.

\iftoggle{FV}{
\subsection{\SYSNAME{} in a Multi Data-Center Cluster}
}
{
\subsection{\SYSNAME{} in a Multi Data-Center Cluster}
}
We also tested \SYSNAME{} in a geo-distributed setting with nodes spread around the world. 
The nodes were placed, one node per region, by the following order, in Amazon's Tokyo, Central, Frankfurt, Paris, Sau-Paulo, Oregon, Singapore, Sydney, Ireland and Ohio data-centers.
We tested only fault free scenarios.
Thus, we kept using BFT-SMaRt rather than its geo-distributed optimized version named WHEAT~\cite{WHEAT}. 
\subsubsection{\SYSNAME's Throughput}
As before, we first measured the bps rate for $n \in \{4, 7, 10\}$.
\begin{figure}[t]
	\centering
	\iftoggle{SINGLE}
	{\includegraphics[width=.49\linewidth]{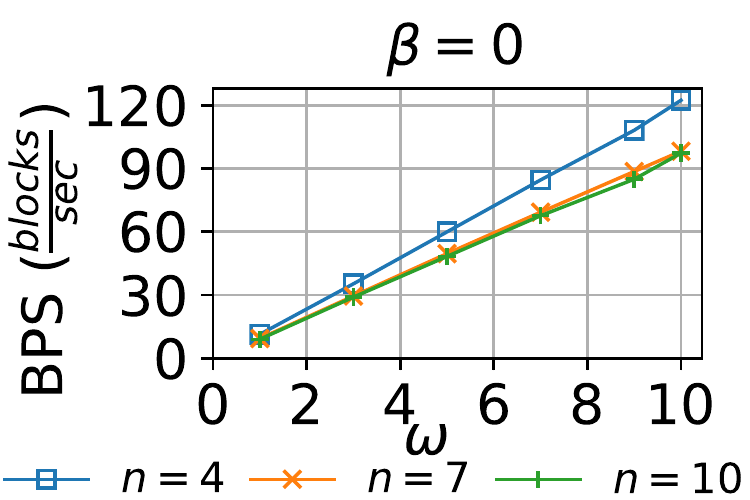}}
	{\includegraphics[width=.55\linewidth]{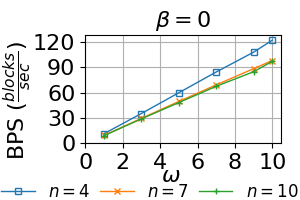}}
	\caption{\SYSNAME's bps rate for $n \in \{ 4, 7, 10\}$ in a multi data-center cluster}
	\label{fig:gdbps}
\end{figure}
Figure~\ref{fig:gdbps} depicts the bps rate for varying cluster sizes. 
As expected, due to lower network performance, the bps is less than $10\%$ of its rate in single data-center clusters.

As in the single data-center deployment, we simulated high load by creating random transactions and run \SYSNAME{} with the following configurations: $n \in \{4, 7, 10\}, 1 \leq \omega \leq 10, \sigma = 512, \beta \in \{10, 100, 1000\}$.
Every benchmark was run for 3 minutes.
\begin{figure}
	\centering
	\iftoggle{SINGLE}
	{\includegraphics[width=0.9\linewidth]{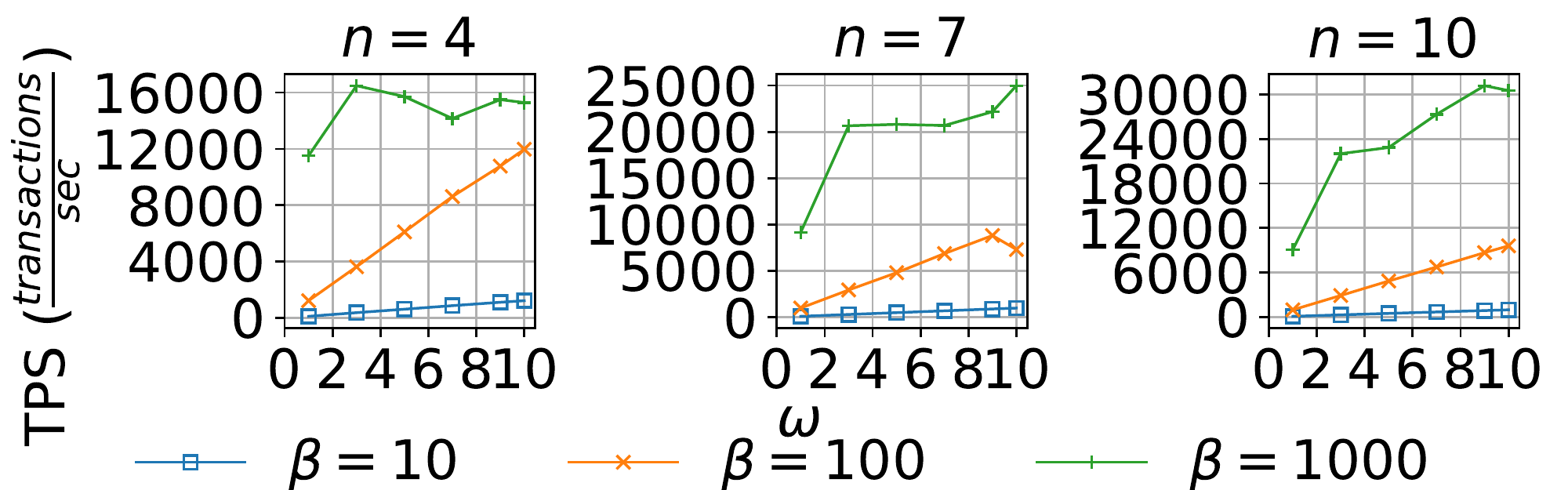}}
	{\includegraphics[width=1\linewidth]{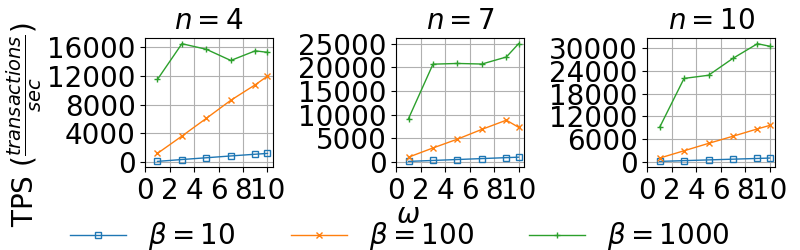}}
	\caption{\SYSNAME's transactions throughput for $\sigma=512$ under various configurations in a multi data-center}
	\label{fig:gdtps}
\end{figure}
Figure~\ref{fig:gdtps} presents the benchmarks result.
Obviously, tps is increasing with $w$ and $\beta$ due to a better CPU utilization and a batching effect.
As for the increase of the tps with $n$, this is, as before, thanks to the separation of blocks from headers which allows bigger clusters to collect more blocks to decide on, thereby enhancing performance.

\iftoggle{FV}{
%Figure~\ref{fig:heatmapbw} presents the respective latencies and obtainable bandwidth between these data-centers as measured by the standard iperf and qperf Linux utilities.
%\begin{figure}
%	\centering
%	\includegraphics[width=1\linewidth]{draws/heatmap_bw}
%	\caption{Heatmap that describes the measured network's bandwidth and latency.
%			The lower triangle shows the bandwidth (in MB/sec) and the upper triangle shows the latency (in milliseconds).}
%	\label{fig:heatmapbw}
%\end{figure}
}{}

%To evaluate \SYSNAME's performance, we first tested the cluster's network bandwidth and latency. 
 
%\begin{figure}
%	\centering
%	\includegraphics[width=1\linewidth]{draws/GD_throughput.pdf}
	%	\begin{fcap}
%	\caption{\SYSNAME's throughput in a geo-distributed settings with the following configurations: $n \in \{4, 7, 10\}, 1 \leq \omega \leq 20, \sigma = 512, \beta \in \{100, 1000\}$ and $\tau = 20000$.}
%	\label{fig:gdthroughput}
	%	\end{fcap}
%\end{figure}

\subsubsection{\SYSNAME's Latency} 
We tested \SYSNAME's latency in the above multi data-center deployment.
As before, we measured the time that takes a block to be delivered by \SYSNAME{} from the moment it was proposed by its creator.
To avoid outliers, we omitted the 5$\%$ most extreme results. 
\begin{figure}[t]
	\centering
	\iftoggle{SINGLE}
	{\includegraphics[width=0.9\linewidth]{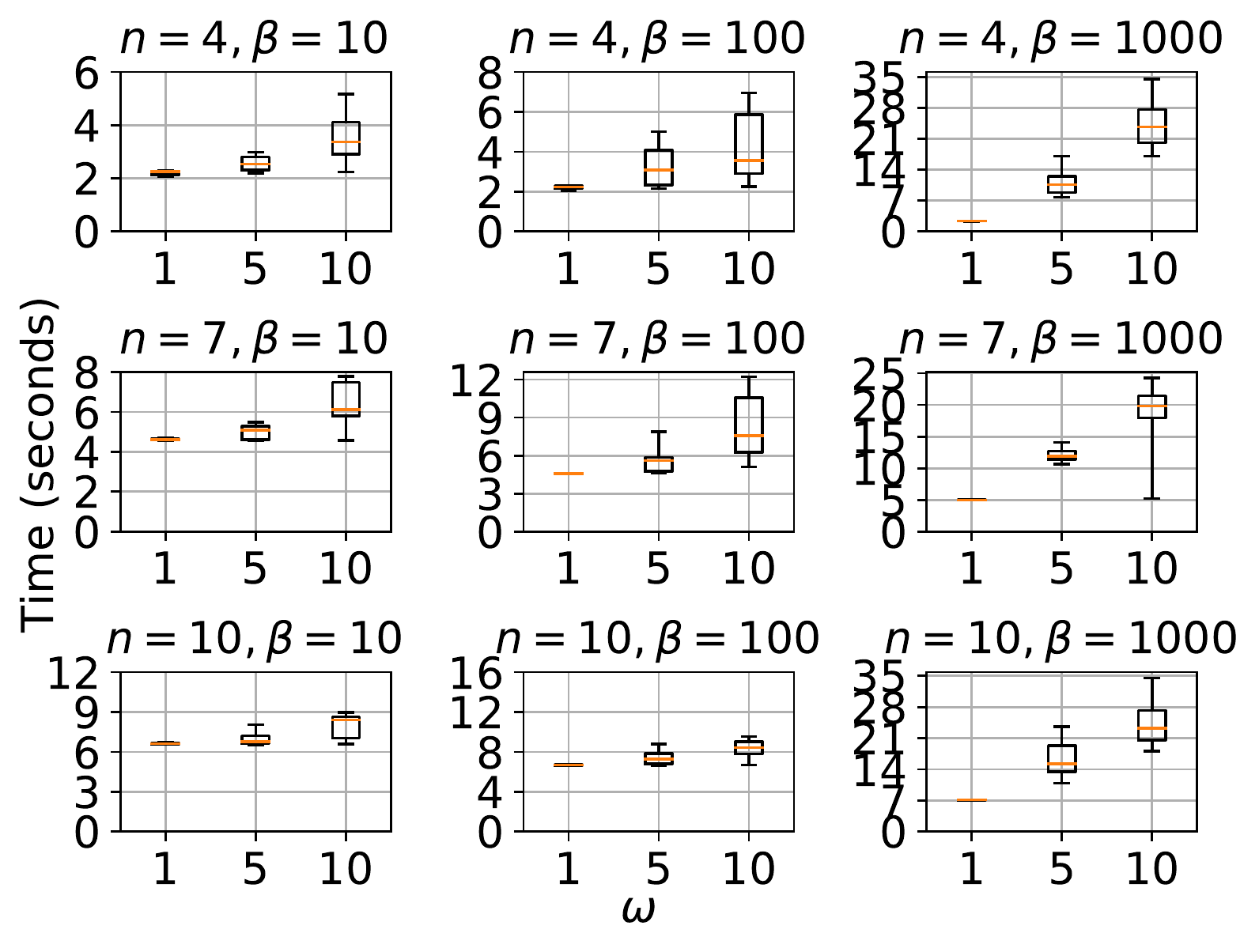}}
	{\includegraphics[width=1\linewidth]{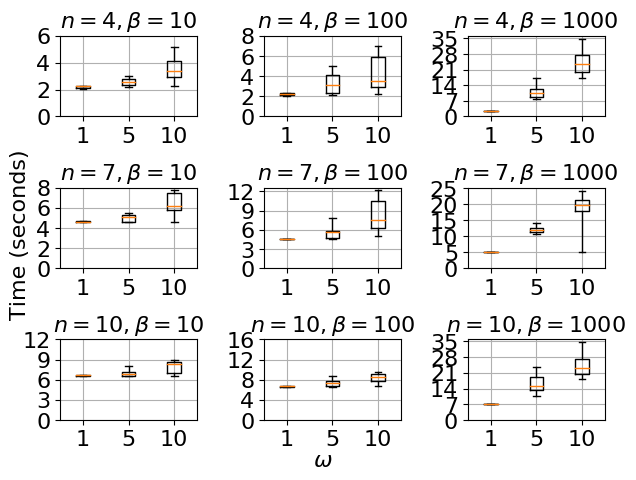}}
	\caption{\SYSNAME's latency in a multi data-center deployment for $\sigma=512, \omega\in \{1, 5, 10\}, \beta\in\{10, 100, 1000\}$ and $n\in\{4, 7, 10\}$}
	\label{fig:gdlatency}
\end{figure}
Figure~\ref{fig:gdlatency} presents the benchmarks results.
For small blocks, the cluster size slightly affects the latency because the headers dissemination process dominates the bandwidth.
Yet, with large blocks, as the data dissemination process itself dominates the latency, increasing the cluster size has very little effect on the latency. 
%\iftoggle{FV}{
%\subsection{Optimizations}
%\label{impl:optimiztions}
%\PROTNAME{} on its own does not obtains an extraordinary performance due to its rotating leader paradigm.
%\SYSNAME{} addresses this issue by running multiple instances of \PROTNAME.
%However, another plausible approach is to disseminate the blocks independently of their headers.
%Namely, a node disseminates through WRB-broadcast only the block header and votes for delivering only if it received the block itself as well.  
%This method, combined with \SYSNAME{} may achieve hundreds of thousands of transactions per second regardless of the block size.
 
%Figures~\ref{fig:cdf},~\ref{fig:heatmap} and~\ref{fig:gdcdf} highlight the main drawback of \SYSNAME. 
%Using multiple workers, in a round robin fashion, may impose higher latency even when a single worker delays.
%Thus, a sharding model in which every worker handles a pre-defined range of possible transaction could alleviate this.
%Cross shards transactions would then be handled by their own dedicated worker to enforce their total ordering among all nodes.  
%}
%{
%}

\subsection{\PROTNAME{} vs. Leading Alternatives}
\label{sec:alternatives}

To the best of our knowledge, the current best performing alternative to \PROTNAME{} is HotStuff~\cite{hotstuff}.
As we had no access to the codebase of HotStuff, we compare the declared performance of HotStuff from~\cite{hotstuff} with our own measurements of \PROTNAME{} using the exact same environment as in~\cite{hotstuff}, namely c5.4xlarge AWS machines (16 vCPUs, 32 GiB RAM).
We also compare with the numbers obtained for BFT-SMaRt~\cite{Bessani2014}, the previous state-of-the-art system, in this same setting.
\begin{figure}[t]
	\centering
	\includegraphics[width=1\linewidth]{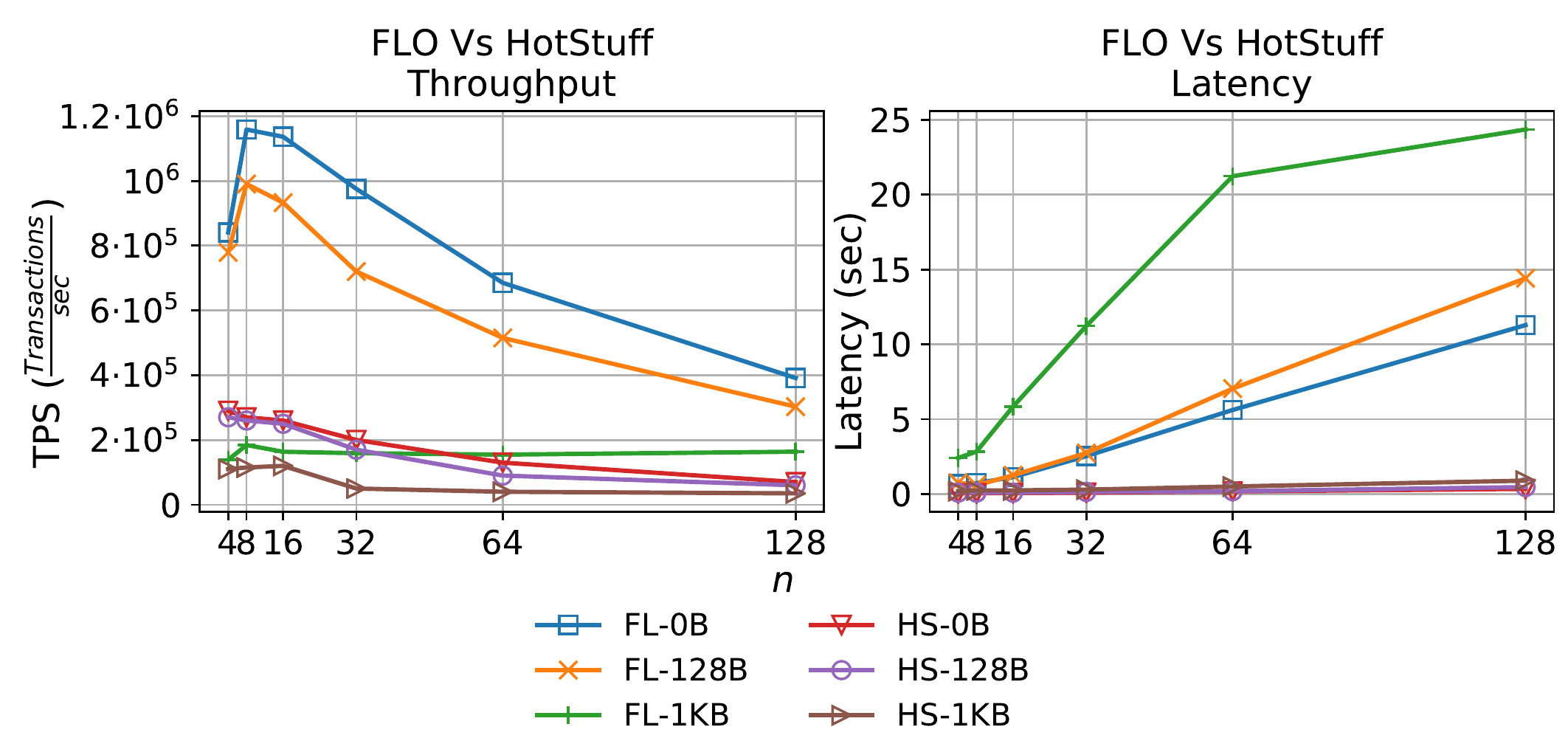}
	\caption{Comparison of \SYSNAME{} and HotStuff on c5.4xlarge AWS machines and $f=\lfloor n/3 \rfloor - 1$}
	\label{fig:tpscs}
\end{figure}
\begin{figure}[t]
	\centering
	\includegraphics[width=1\linewidth]{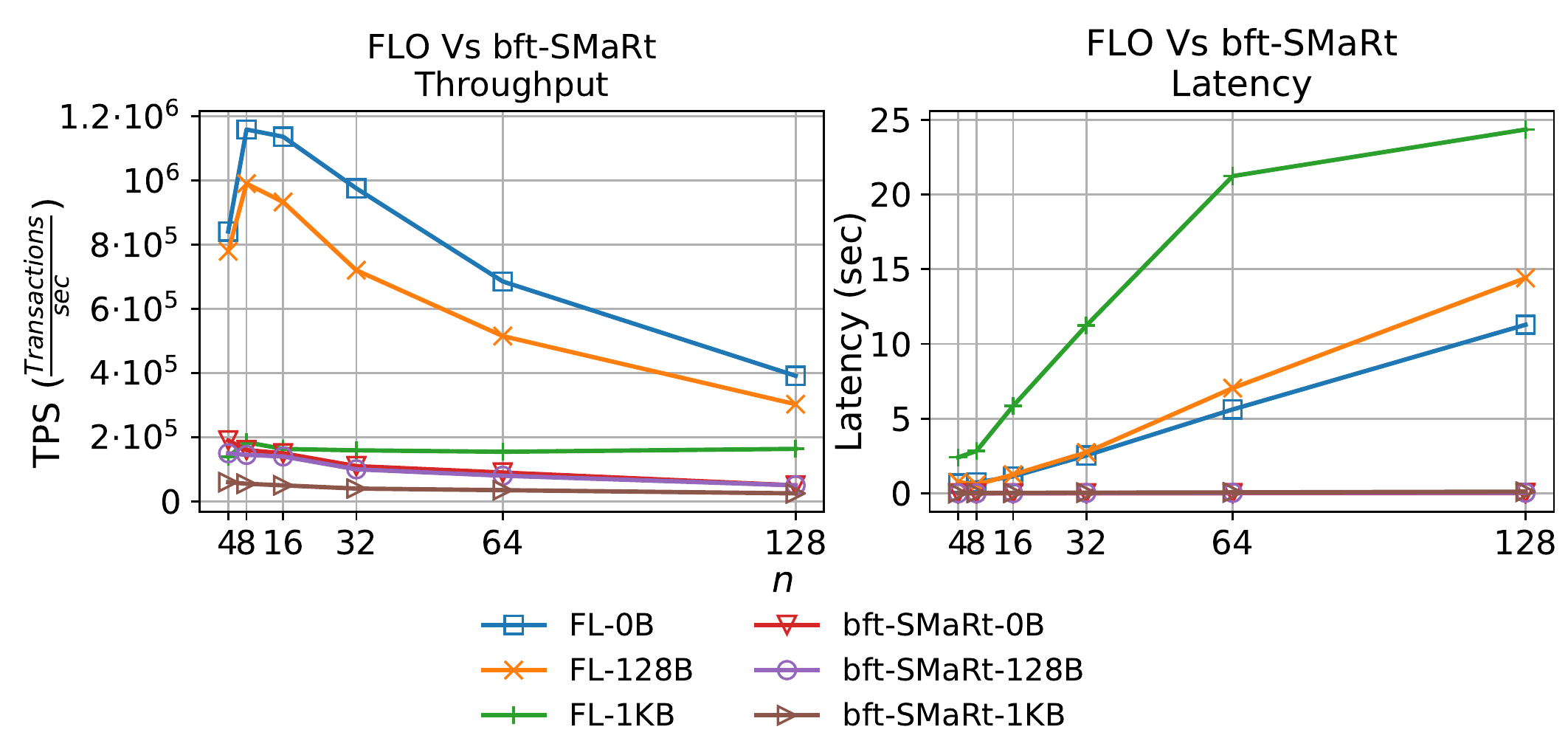}
	\caption{Comparison of \SYSNAME{} and bft-SMaRt on c5.4xlarge AWS machines and $f=\lfloor n/3 \rfloor - 1$}
	\label{fig:tpscsbfts}
\end{figure}
Figure \ref{fig:tpscs} and \ref{fig:tpscsbfts} show this comparison's results w.r.t the $n$ and $\sigma$ parameters.
For all cluster sizes \SYSNAME{} was deployed with $\beta=1000$ and $\omega=8$ with maximal resiliency $f=\lfloor n/3 \rfloor - 1$. 
In terms of throughput, for any $\sigma$ and $n$ \SYSNAME{} performs $20\%-300\%$ better than HotStuff and $40\%-600\%$ better than bft-SMaRt.
Notice that HotStuff is implemented in C while \SYSNAME{} is in Java.

As for latency, due to the $f+1$ finality of \PROTNAME{} and the fact that we run with maximal resiliency, \SYSNAME's latency rises with to $n$.
In contrast, transactions' finality with both HotStuff and bft-SMaRt is at most three rounds (HotStuff), so their latency is much less impacted by $n$.
Still, in all cases the latency obtained by \SYSNAME{} is better than most existing cryptocurrencies (and tokens), including the very recent Algo (Algorand)~\cite{Gilad:2017}. 
Libra's target 10 seconds finality~\cite{libra-bc} is met by \SYSNAME{} when $n \leq 30$ nodes\footnote{Ripple and Stellar, for example, run on smaller clusters.}.

\paragraph*{Conclusions} The performance gap between \SYSNAME{} and the alternatives narrows as transactions become larger.
This is because in such cases, the basic need to disseminate the transactions dominates the communication overhead, so a clever consensus protocol has less room for impact.
Hence, one should consider compressing the data for large transactions.
Also, by employing scalable Byzantine dissemination protocols for the transactions data, e.g.,~\cite{DFKS07}, \SYSNAME{} is likely to better handle large clusters.

%The results are listed in Table~\ref{tbl:flvshs}.
%As can be seen, \PROTNAME{} obtains between $20\%-300\%$ higher throughput than HotStuff depending on the transactions size despite the fact that HotStuff is written in C++ whereas \PROTNAME{} is written in Java.
%Compared to BFT-SMaRt, \PROTNAME{} obtains between $2.5$ to $5$ times better throughput.
%A more elaborate comparison would be included in the next version of this paper.

%\begin{table}[t]
%	\centering
%	\begin{tabu} to 1\linewidth {||r|r|r|r||}
%		\hline
%		\makecell{TX size (bytes)} & \makecell{\PROTNAME} & \makecell{HotStuff} & \makecell{BFT-SMaRt} \\ 
%		\hline
%		\hline
%		0 & 1M & 310K & 190K \\
%		\hline
%		128 & 800K & 230K & 150K \\
%		\hline
%		1024 & 145K & 120K & 60K \\
%		\hline
%	\end{tabu}
%	\caption{Max TPS for \PROTNAME{} vs. HotStuff and BFT-SMaRt on c5.4xlarge AWS machines with $n=4$ for various transaction sizes (in this setting \PROTNAME{} obtains $280,000$ TPS with $512$B transactions)}
%	\label{tbl:flvshs}
%\end{table}

	\section{Discussion}
\label{discussion}
\PROTNAME{} is a communication frugal optimistic block\-chain algorithm targeting environments where failures rarely occur.
For example, \PROTNAME{} is likely to be very attractive for the FinTech industry, which uses highly secure and robust systems.
\PROTNAME{} leverages blockchain's iterativity as well as its cryptographic features to achieve its goal.

Loosely speaking, \SYSNAME{} employs \PROTNAME{} as a block\-chain-based consensus algorithm rather than consensus-based blockchain.
Our performance results show that it matches the requirements of real demanding commercial applications even when executing on common non-dedicated infrastructure.
\nottoggle{VLDB}{

In this paper we studied the distributed agreement aspects of \SYSNAME/\PROTNAME's optimistic approach.
In the future, we intend to examine other aspects such as sharding, data model, scalability, and improvement of \PROTNAME's latency.}{In the future, we intend to explore sharding, which can potentially give an additional significant performance boost, as well as scalable dissemination protocols.}
\iftoggle{Annon}{
Our prototype implementation of this work is available in open source~\cite{NugasDSL-anon}.
}{
Finally, our prototype implementation of this work is available in open source~\cite{NugasDSL}.
}

\nottoggle{VLDB}{
\iftoggle{Annon}{}{
\begin{center}
	\textit{$\sim$Expect the best. Prepare for the worst.$\sim$}\\	
	\textit{Zig Ziglar}
\end{center}
}
}{
}

	\clearpage
\bibliographystyle{abbrv}
\bibliography{references}
\clearpage
\appendix
%\begin{appendices}
	\section{Optimistic Binary Byzantine Consensus}
\label{app:obbc}
\subsection{Introduction}
The Optimistic Binary Byzantine Consensus (OBBC) is a key abstraction in WRB's ability to deliver messages in a single communication step in favorable conditions, which is the basis for \PROTNAME's low latency and frugal communication pattern (see lines~\ref{wrb:l2}--\ref{wrb:l5} in Algorithm~\ref{alg:wrb}).
Recall that the favorable conditions for \PROTNAME{} are that there is no Byzantine activity and the network delivers the proposer's messages in a timely manner.
If these conditions are met, all nodes receive the proposer's message (line~\ref{wrb:l3} in Algorithm~\ref{alg:wrb}).
Consequently, they will all invoke OBBC with $1$.
In other words, the implementation of OBBC must terminate in a single round whenever all nodes vote $1$.
This enables us to develop a deterministic protocol that withstands $f < \frac{n}{3}$ byzantine nodes.
This is in contrast to other existing OBBC implementations, e.g.~\cite{Friedman:2005, Mosfaoui:2015}, that either rely on an oracle, randomization, or weaker failure model.
To that end, we denote a BBC protocol that terminates quickly when all nodes vote for a given value $v$ by OBBC$_v$;
specifically, Algorithm~\ref{alg:wrb} invokes OBBC$_1$.

%The specific implementation has a significant impact on \PROTNAME's performance in the presence of failures. 

\begin{comment}
In this section, we present our OBBC implementation. 
Other algorithms, e.g.~\cite{Friedman:2005, Mosfaoui:2015}, often rely on an oracle, randomization, or weaker failure model.
In contrary, our algorithm is deterministic and supports $f < \frac{n}{3}$ byzantine nodes. 
In exchange, we restrict the set of cases in which the algorithm decides in a single communication step for those cases that are in the primary concern of \PROTNAME{} \textbf{[[[What are they?]]]} which are periods of synchronization, no byzantine nodes, and a correct proposer.

Although our algorithm is tailored for \PROTNAME's needs, it may be helpful for other problems that share \PROTNAME's optimistic assumptions.
\end{comment}

\subsection{OBBC$_v$ Formal Definition}
%We denote OBBC$_v$ an OBBC that may decide $v$ in a single communication step if a predefined set of conditions is met. 
%This is in accordabe with Section~\ref{sec:prelim}
%
\textbf{OBBC$_v$-Agreement} and \textbf{OBBC$_v$-Termination} are the same as their BBC counterparts. 
In addition we define the following:
\iftoggle{MYACM}{
	\begin{description}
	}{	
		\begin{LaTeXdescription}
		}
	\item[evidence($v$):] %An evidence that convinces a correct node to vote for $v$. 
	%Formally, if a correct node $p$ has received a valid $evidence(v)$, then $p$ adopts $v$. \textbf{[[[what does it mean to adopt $v$]]]}
	%Adopting a value $v$ means that if a process has not yet decided, from the adoption point and on, it proposes $v$ (even if $v$ was not its original proposal).
	%
	An $\mathit{evidence}(v)$ for a value $v$ is a cryptographic proof that can be verified by an external valid function.
	For a value $v'\neq v$, $\mathit{evidence}(v) = nil$.
	%We assume that evidence($v$) cannot be forged and its validation is defined by a known reliable mechanism (e.g. cryptographic mechanisms).     
	%\item[OBBC$_v$-Validity:] A decided value $v'$ was either proposed or adopted at some point by a correct node.
	\item[OBBC$_v$-Validity:] A value $v'$ decided by a correct node was either proposed by a correct node or $v'=v$ and it was proposed by a node that has a valid $\mathit{evidence}(v)$.
	\item[OBBC$_v$-Fast-Termination:] If no node has proposed $v'\neq v$, then $v$ will be decided by every correct node in a single communication step.
	\iftoggle{MYACM}{
\end{description}
}{	
\end{LaTeXdescription}
}

\subsection{Implementing OBBC$_v$}
\begin{algorithm}[t]
	\caption{OBBC$_v$ - code for $p$}
	\label{OBBC:impl}
    \setcounter{AlgoLine}{0}
	\SetNlSty{texttt}{OB}{}
	\footnotesize
	\SetKwFunction{PROP}{BBC$_v$.propose}
	\SetKwFunction{OPROP}{OBBC$_v$.propose}
	\SetKwFunction{BR}{broadcast}
	\SetKwFunction{Send}{send}
	\SetKwFunction{WRBB}{WRB-broadcast}
	\SetKwFunction{WRBD}{WRB-deliver}
	\SetKwFunction{ASSERT}{assert}
	\SetKwRepeat{Wait}{wait}{until}
	\SetKwFor{Upon}{upon}{do}{end} 
	\SetKwProg{myproc}{Procedure}{}{}
	\footnotesize
%	\begin{algorithmic}[1]
	\myproc{\OPROP{$v'$, $\mathit{evidence}(v)$}}{
		\ASSERT{$v = v' \implies \mathit{evidence}(v)$ is valid} \;
		\ASSERT{$v \neq v' \implies \mathit{evidence}(v) = nil$} \;
		\BR{v'} \label{OBBC:2} \;
		\lWait{$n-f$ proposals $\hat{v}$ have been received \label{OBBC:3}}{}
		$\mathit{votes}=$ \{received proposals\} \;
		\If {$\mathit{votes}=\{v\}$} { 
			decide $v$ \label{OBBC:5}\;
			return $v$\;
		} 
		$\mathit{evidences}=\{\}$ \tcc{couldn't terminate quickly; run full protocol}
		%\State $broadcast(EV,v)$ \label{OBBC:7} \;
		\BR{$EV$} \label{OBBC:7} \;
		\lWait{$|\mathit{evidences}| = n-f$ \label{OBBC:8}}{} \tcc{see also line~OB\ref{OBBC:18}}
		$\mathit{new\_v} = v'$ \label{OBBC:9}\;
		\If{$\mathit{evidences}$ contains a valid $\mathit{evidence}(v)$} {
			$\mathit{new\_v} = v$ \tcc{notice: only $v$ can have a valid evidence}
		}
		return \PROP{$\mathit{new\_v}$} \label{OBBC:12} \tcc{invoke a BBC protocol}
	}
		\BlankLine
		%\State \textbf{upon} receiving $(EV)$ from $q \land \mathit{evidence}(v) \neq nil$ \textbf{do}: \label{OBBC:14}  
		\Upon{receiving $(EV)$ from $q$ \label{OBBC:14}}
		{
			\Send{$\mathit{evidence}(v)$) to $q$} \label{OBBC:15}
		}
		\BlankLine
		%\State \textbf{upon} receiving a valid $\mathit{evidence}(\hat{v})$ from $q$ \textbf{do}: \label{OBBC:17}
		\Upon{receiving an $\mathit{evidence}(\hat{v})$ from $q$ \label{OBBC:17}}
		{
			$\mathit{evidences}= \mathit{evidences} \cup \{\mathit{evidence}(\hat{v})\}$ \label{OBBC:18}
		}
		\BlankLine
		\If {invoked Decide $v$ and some node invoked BBC.propose($v'$)} {
			\PROP{$v$} \tcc{this is executed at most once}
		}
%	\end{algorithmic}
\end{algorithm}
The pseudocode implementation of OBBC$_v$ is listed in Algorithm~\ref{OBBC:impl}. 
\emph{OBBC$_v$.propose} receives two parameters: the actual proposal $v'$ and an $evidence(v)$.
If $p$ is correct and $v=v'$, then $evidence(v)$ is valid.
Else, if $v \neq v'$ then $evidence(v)$ is $nil$.

When $p$ invokes \emph{OBBC$_v$.propose($v'$, $evidence(v)$)}, it performs the following:
\begin{itemize}
	\item Broadcast $v'$ to all (line~OB\ref{OBBC:2}).
	\item Waits for at most $n-f$ proposals. If a single value $v'=v$ has been received, then $p$ decides $v$ (lines OB\ref{OBBC:3}--OB\ref{OBBC:5}).
	\item Else, $p$ broadcast a request for $evidence(v)$ and waits for $n-f$ replies (which might be $nil$) (lines OB\ref{OBBC:7}--OB\ref{OBBC:8} and lines OB\ref{OBBC:17}--OB\ref{OBBC:18}).
	\item In addition, if $p$ receives a request for $evidence(v)$ from $q$ and $p$ has a valid one, $p$ sends back that evidence to $q$ (lines OB\ref{OBBC:14}--OB\ref{OBBC:15}).
	\item Finally, if $p$ has received a valid $evidence(v)$ it adopts $v$. Then $p$ proposes its value through a regular BBC (lines OB\ref{OBBC:9}--OB\ref{OBBC:12}).
\end{itemize}

\subsection{OBBC$_v$ Correctness proof}
\begin{lemma}
	\label{OBBC:lemmas:1}
	If a correct node $p$ has decided $v$ in line~OB\ref{OBBC:5}, then all correct nodes who did not decide in line~OB\ref{OBBC:5} set $new\_v=v$.
\end{lemma}
\begin{proof}
	By Algorithm~\ref{OBBC:impl}, if $p$ is correct and has decided at line~OB\ref{OBBC:5} then at least $f+1$ correct nodes broadcast $v$.
	Thus, at least $f+1$ correct nodes have a valid $evidence(v)$.
	As for $v' \neq v$, $\mathit{eveidence}(v)=nil$ and by the ratio between $n$ and $f$, a correct node at line~OB\ref{OBBC:8} will receive at least one valid $evidence(v)$ and thus set $new\_v = v$.
\end{proof}
\begin{lemma}
	\label{OBBC:lemma:2}
	If a correct node $p$ has decided $v$ at line~OB\ref{OBBC:5}, then all correct nodes will eventually decide $v$.
\end{lemma}
\begin{proof}
	By Lemma~\ref{OBBC:lemmas:1}, any correct node who did not decide at line~OB\ref{OBBC:5} sets $new\_v=v$.
	By \textbf{BBC-Validity}, if all correct nodes invoke BBC.propose($v$) with the same value $v$ then $v$ must be decided.
\end{proof}
\begin{lemma}[OBBC$_v$-Agreement]
	\label{OBBC:lemma:agreement}
	No two correct nodes decide differently.
\end{lemma}
\begin{proof}
	If any correct node has decided $v$ at line~OB\ref{OBBC:5}, then by Lemma~\ref{OBBC:lemma:2} all correct nodes will eventually decide $v$.
	Else, all correct nodes decide at line~OB\ref{OBBC:12} and by \textbf{BBC-Agreement} no two correct nodes decide differently.
\end{proof}
\begin{lemma}[OBBC$_v$-Validity]
	\label{OBBC:lemma:validity}
	A value $v'$ decided by a correct node was either proposed by a correct node or $v'=v$ and it was proposed by a node with a valid $\mathit{evidence}(v)$.   
\end{lemma}
\begin{proof}
	Assume b.w.o.c. that $v'$ has been decided by a correct node but had not been proposed by any correct node and no correct node has received a valid $evidence(v')$.
	If no correct node proposed $v'$, it means that all correct nodes executed lines~OB\ref{OBBC:8}--OB\ref{OBBC:12} and all invoked BBC.propose with their initial value.
	Further, since none of them had a valid $evidence(v')$, their initial value is $v'\neq v$, and therefore the same (there are only two possible values in BBC).
	Hence, by \textbf{BBC-Validity} all correct nodes decided $v'$ which is also the value proposed by all of them.
	A contradiction. 
\end{proof}
\begin{lemma}[OBBC$_v$-Termination]
	\label{OBBC:lemma:termination}
	Every correct node eventually decides.
\end{lemma}
\begin{proof}
	By the ratio between $n$ and $f$, no correct node blocks forever at lines~OB\ref{OBBC:3} or~OB\ref{OBBC:8}. 
	By \textbf{BBC-Termination} no correct node is blocked forever at line~OB\ref{OBBC:12}.
	Thus, every correct node eventually decides.
\end{proof}
\begin{lemma}[OBBC$_v$-Fast-Termination]
	\label{OBBC:lemma:ftermination}
	If no process has proposed $v'\neq v$, then $v$ will be decided by every correct node in a single communication step.
\end{lemma}
\begin{proof}
	By the ratio between $n$ and $f$ and by lines~OB\ref{OBBC:3}--OB\ref{OBBC:5} a correct nodes will receive $n-f$ proposals of $v$ and thus decide $v$ in a single communication step.
\end{proof}
\begin{theorem}
	By Lemmas~\ref{OBBC:lemma:agreement},~\ref{OBBC:lemma:validity},~\ref{OBBC:lemma:termination} and~\ref{OBBC:lemma:ftermination}, Algorithm~\ref{OBBC:impl} solves the OBBC$_v$ problem.
\end{theorem}

\subsection{OBBC$_v$ and WRB}
In the context of WRB, we use \emph{OBBC$_1$.propose} with $\mathit{evidence}(1)=(m, sig_{proposer}(m))$ as an evidence to run the \emph{OBBC.propose} procedure at lines~\ref{wrb:l4}--\ref{wrb:l5} of Algorithm~\ref{alg:wrb}. 
As we expect to face mostly Byzantine failure free synchronized executions, it is beneficial to use OBBC$_1$.
Even in the presence of benign failures, as long as the proposer is correct and messages from correct nodes arrive in a timely manner, OBBC$_1$ is still able to terminate fast in a single communication step.

%\end{appendices}
\end{document}